\documentclass[conference]{IEEEtran}
\usepackage{graphicx}
\usepackage{balance}  
\usepackage{xcolor}
\usepackage{algorithmic}
\usepackage{defpaper2e}
\usepackage{amsmath}
\usepackage{url}
\usepackage{float}
\usepackage{multirow}
\usepackage[boxruled, linesnumbered] {algorithm2e}
\usepackage{parskip}
\usepackage{subcaption}
\usepackage{kbordermatrix}
\usepackage{mathtools, cuted, lipsum, tikz}
\usepackage{dsfont}
\usepackage{amsthm}
\renewcommand{\kbldelim}{(}
\renewcommand{\kbrdelim}{)}
\DeclareMathOperator*{\argmin}{arg\,min}

\newtheorem{thm}{Theorem}
\newtheorem{defi}[theorem]{Definition}

\newcommand{\Boxiang}[1]{{{\textcolor{black}{\textbf{Boxiang:}}}{\textcolor{blue}{\textbf{#1}}}}}
\newcommand{\Wendy}[1]{{{\textcolor{black}{\textbf{Wendy:}}}{\textcolor{red}{\textbf{#1}}}}}
\newcommand{\Haipei}[1]{{{\textcolor{black}{\textbf{Haipei:}}}{\textcolor{violet}{\textbf{#1}}}}}
\newcommand{\Ting}[1]{{{\textcolor{black}{\textbf{Ting:}}}{\textcolor{green}{\textbf{#1}}}}}

\newcommand{\comment}[1]{}

\begin{document}

\title{Truth Inference on Sparse Crowdsourcing Data with Local Differential Privacy}



%
%
%
%

\author{
\IEEEauthorblockN{Haipei Sun\IEEEauthorrefmark{1},
Boxiang Dong\IEEEauthorrefmark{2},
Hui (Wendy) Wang\IEEEauthorrefmark{1},
Ting Yu\IEEEauthorrefmark{3} and
Zhan Qin\IEEEauthorrefmark{4}}
\IEEEauthorblockA{\IEEEauthorrefmark{1}Stevens Institute of Technology, Hoboken, New Jersey 07030\\
Email: \{hsun15, Hui.Wang\}@stevens.edu}
\IEEEauthorblockA{\IEEEauthorrefmark{2}Montclair State University, Montclair, New Jersey 07043\\
Email: dongb@montclair.edu}
\IEEEauthorblockA{\IEEEauthorrefmark{3}Qatar Computing Research Institute, Doha, Qatar\\
Email: tyu@qf.org.qa}
\IEEEauthorblockA{\IEEEauthorrefmark{4}The University of Texas at San Antonio, San Antonio, Texas 78249\\
Email: zhan.qin@utsa.edu}}



\maketitle

\begin{abstract}
Crowdsourcing has arisen as a new problem-solving paradigm
for tasks that are difficult for computers but easy for humans. However, since the answers collected from the recruited participants (workers) may contain sensitive information, crowdsourcing raises serious privacy concerns. 
In this paper, we investigate the problem of protecting answer privacy under {\em local differential privacy} (LDP), by which individual workers randomize their answers independently and send the perturbed answers to the task requester. 
The utility goal is to enable to infer the true answer (i.e., truth) from the perturbed data with high accuracy. 
One of the challenges of LDP perturbation is the {\em sparsity} of worker answers (i.e., each worker only answers a small number of tasks). 
Simple extension of the existing approaches (e.g., Laplace perturbation and randomized response) may incur large error of truth inference on sparse data. 
Thus we design an efficient new matrix
factorization (MF) algorithm under LDP. 
We prove that our MF algorithm can provide both LDP guarantee and small error of truth inference, regardless of the sparsity of worker answers.
We perform extensive experiments on real-world and synthetic datasets, and demonstrate that the MF algorithm performs better than the existing LDP algorithms on sparse crowdsourcing data.
\end{abstract}

\IEEEpeerreviewmaketitle

\vspace{-0.1in}
\section{Introduction}
\label{sc:intro}
\vspace{-0.1in}
Crowdsourcing enables to perform tasks that are easy for humans but remain difficult for computers. 
Typically, a {\em task requester} releases the tasks on a crowdsourcing platform (e.g., Amazon Mechanical Turk (AMT)). Then the {\em workers} provide their answers to these tasks on the same  crowdsourcing platform in exchange for a reward. 
As the Internet and mobile technologies continue to advance, crowdsourcing helps organizations to improve productivity and innovation in a cost-effective way. 
 
While crowdsourcing provides an effective way for problem-solving, collecting individual worker answers may pose potential privacy risks on the workers. For instance, it has been reported that AMT was leveraged by politicians to access a large pool of Facebook profiles and collect tens of thousands of individuals' demographic data \cite{tedcruz}. Crowdsourcing-related applications such as participatory sensing \cite{burke2006participatory} and citizen science \cite{bowser2014sharing} also raise privacy concerns for  the workers. 
However, simply removing the worker names or replacing the names with pseudo-names cannot adequately protect worker privacy. 
It is still possible to de-anonymize crowd workers by matching their inputs with external datasets \cite{kandappu2015privacy}. 

Recently, differential privacy (DP) \cite{dwork2008differential} has been used in many applications to provide rigorous privacy guarantee. Classical DP requires a centralized trusted {\em data curator} (DC) to collect all the answers and publish privatized statistical information. 
However, as crowdsourcing services become increasingly popular, many untrusted task requesters (or DCs) appear to abuse the crowdsourcing services by collecting private information of the workers (see the aforementioned example that AMT was leveraged by politicians). 
In recent years, local differential privacy (LDP) \cite{duchi2013local} arises as a good alternative paradigm to prevent an untrusted DC from learning any personal information of the data providers and provides {\em per-user privacy}. In LDP, each data provider randomizes his/her data to satisfy DP locally. Then the data providers send their randomized answers to the (untrusted) DC, who aggregates the data with no access to personal information of data providers. LDP roots in {\em randomized response}, first proposed in \cite{warner1965randomized}.
It has been used in many applications such as Google's Chrome browser \cite{erlingsson2014rappor,fanti2016building} and Apple's iOS 10 \cite{appledp}. 

The goal of this paper is to protect worker privacy  with LDP guarantee for crowdsourcing systems. We aim to address two main challenges. 

\noindent{\bf Challenge 1: sparsity of worker answers.} The worker answers can be very {\em sparse}, as often times most workers only provide answers to a very small portion of the tasks. For example, the analysis of a real-world crowdsourcing dataset, named {\em AdultContent} dataset\footnote{http://dbgroup.cs.tsinghua.edu.cn/ligl/crowddata}, which includes the relevance ratings of 825 workers for approximately 11,000 documents, shows that the dataset has the average sparsity as high as 99\% (i.e., each worker rates 70 documents at average). 
The NULL values in worker answers (i.e., missing answers from the workers) should be protected, as they reveal the sensitive information as whether a worker has participated in specific tasks  \cite{ciglic2014k}. However, careless perturbation of NULL values may alter the original answer distribution and incur significant inaccuracy of truth inference. However, most of the LDP works (e.g., \cite{bassily2017practical,bassily2015local,qin2016heavy,wang2017locally}) only consider complete data. None of these works address the sparsity issue. 

\noindent{\bf Challenge 2: data utility}. One of the major utility goals of crowdsourcing systems is to aggregate the answers from workers of different (possibly unknown) quality, and infer the true answer (i.e., truth) of each task. Truth inference \cite{li2014confidence, zhao2012probabilistic} has been shown as  effective to estimate both worker quality and the truth. Most of the existing truth inference mechanisms follow an iterative process: both worker quality and inferred truth are estimated and updated iteratively, until they reach convergence. This makes it extremely difficult to preserve the accuracy of truth inference on the perturbed worker answers, as even the slight amount of initial noise on the worker answers may be propagated during iterations and amplified to the final truth inference results. 

Matrix factorization is one of the most popular methods to predict NULL values. 
Intuitively, by matrix factorization, the answers of a specific worker can be modeled as the multiplication of a worker profile vector and a task profile matrix. We observe that the adaption of a worker profile vector only relies on the worker's own answers rather than those of other workers. Therefore, matrix factorization can be performed locally. However, applying LDP to matrix factorization is not trivial. 
A direct application of LDP on the worker answers can incur 
a large perturbation error that linearly grows with
the number of tasks. 
Most of the existing methods of matrix factorization with LDP \cite{hua2015differentially,shen2014privacy,shen2016epicrec,shin2018privacy} only consider recommender systems. They cannot be applied to the crowdsourcing setting directly due to the mismatched utility function; they consider the accuracy of recommendation while we consider the accuracy of truth inference. Furthermore, 
some of these works (e.g., \cite{hua2015differentially,shin2018privacy}) require iterative feedbacks between the data curator and users, which may incur expensive communication overhead. 

\noindent{\bf Contributions.} 
To our best knowledge, this is the first work that provides  privacy protection on sparse worker answers with LDP guarantee, while preserving the accuracy of truth inference on the perturbed data. 
We summarize our main contributions as follows. 
First, we present simple extension of two existing perturbation approaches, namely Laplace perturbation (LP) and randomized response (RR), to deal with sparse crowdsourcing data. We provide detailed analysis of the expected error bound of truth inference for these two approaches, and show that both LP and RR can have large error of truth inference for sparse data. 
Second, we design a new matrix factorization (MF) algorithm with LDP. To preserve the accuracy of truth inference, we apply the perturbation on the objective function of matrix factorization, instead of on the matrix factorization results. Our formal analysis shows that the theoretical error bound of truth inference results on the perturbed data is small, and the error bound of truth inference is insensitive to data sparsity. 
Finally, we conduct extensive experiments on real-world and synthetic crowdsourcing datasets of various sparsity. The experimental results demonstrate that our MF approach significantly improves the accuracy of truth inference on the perturbed data compared with the existing work. 

The rest of this paper is organized as follows.
Section \ref{sc:back} and \ref{sc:preliminaries} present the background and preliminaries. Section \ref{sc:straw}  introduces the extension of two existing approaches. Section \ref{sc:mf} discusses our MF mechanism. Section \ref{sc:exp} presents the experiment results. Related work is discussed in Section \ref{sc:related}. Section \ref{sc:conclu} concludes the paper.

\nop{
One possible solution is to apply \emph{differential privacy} (DP) \cite{dwork2008differential}, which provides. 
However, due to the fact that the crowdsourcing data is sparse in general, enforcing DP by adding Laplace noise on many NULL answers will damage the utility of the crowdsourcing data. 
A core {\em utility} function of the crowdsourcing data is {\em truth inference}, i.e., to aggregate and infer the data analytics results from workers' (possibly noisy) answers. Adding noise for privacy protection on those worker answers, which are already noisy, no doubt will downgrade the truth inference results. 


\begin{figure}[ht!]
     \includegraphics[width=\linewidth]{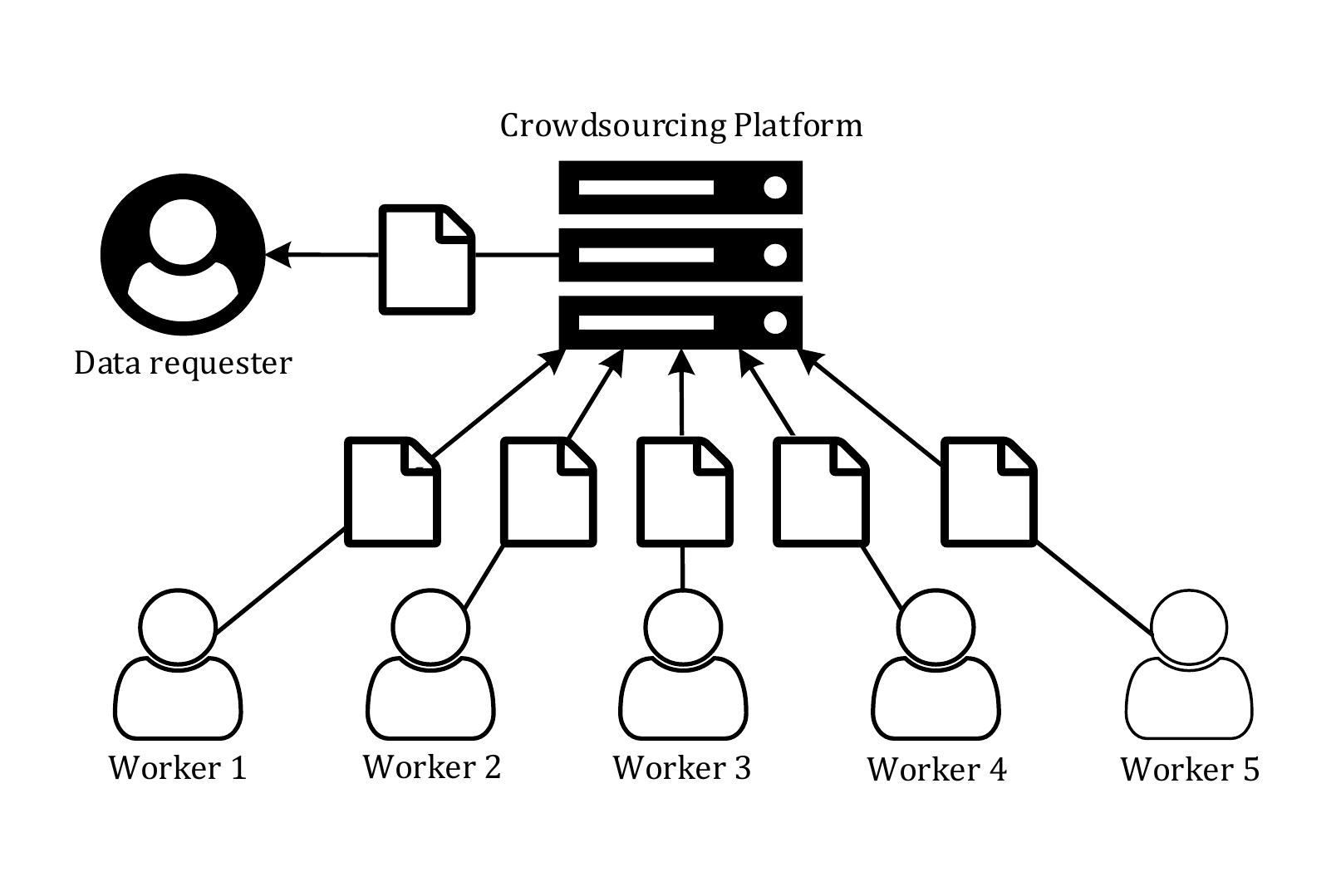}
  \caption{Crowdsourcing Framework}
  \label{fig:framework}
\end{figure}

A typical type of systems that exploit the crowd intelligence to improve the performance is the data analytics systems. Recently, several academic and industrial
organizations have started investigating and developing technologies and infrastructure that supports various crowdsourced data analytics (CDA) systems, including data collection \cite{trushkowsky2013crowdsourced,park2014crowdfill}, search \cite{franklin2011crowddb,parameswaran2011human,parameswaran2012crowdscreen,parameswaran2014optimal,yan2010crowdsearch}, ranking \cite{chen2013pairwise, guo2012so}, data cleaning \cite{wang2012crowder,chu2015katara}, to name a few. 


However due to the setting of crowdsourcing platform, the answer of each user must be protected before uploading to the platform. Here we use another technique called \emph{local differential privacy} (LDP) \cite{kasiviswanathan2011can}. But for both DP and LDP, there is an unsolved problem: since most of the platform allow user to skip tasks, i.e., only finish a subset of the tasks, the answer vector returned to the platform may be sparse and contain NULL values. There are many existing works such as \Haipei{Add citation here.}, but to the best of our knowledge, none of these works has specifically deal with sparsity on non-aggregated data.


\noindent{\bf Challenges.} There are two challenges in this project. First, we need to deal with NULL values in the answer vector. According to the setting of our problem, we want to protect the privacy of the answer vector from each user, thus the data cannot be aggregated. The action of skipping a task is also a kind of privacy, so we need to apply privacy protect on both normal values and NULL value. However, as NULL is a very special value in database, we cannot simply replace it with a special value without giving the formal proof. None of the existing privacy mechanisms have considered about how to formally deal with NULL values, i.e., the sparsity. Second, we need to evaluate the potential quality loss of worker answers due to the privacy protection mechanisms on the tasks. None of the existing privacy preservation mechanisms have considered the accuracy of the truth inference results as the desired utility. Indeed, there are many privacy-preserving mechanisms and each of them has different utility impact truth inference. It is necessary to evaluate the performance of truth inference algorithm using different mechanisms, especially on a sparse dataset with NULL value.

\noindent{\bf Our contributions.} We treat the the answers of each worker as a \emph{answer vector}, in which each cell represents the answer to a particular task. We improved two existing different privacy mechanisms: Laplace Perturbation and Random Response to make them possible to handle NULL value in sparse answer vector and give the formal proof. In addition, we design a new matrix factorization method that can do the same thing and give the proof as well. We use all three approaches above to achieve local differential privacy on sparse dataset. Then we evaluate and compare the utility impact of the approaches on truth inference algorithm.
}

\vspace{-0.1in}
\section{Background}
\label{sc:back}
\vspace{-0.1in}
In this section, we briefly recall the definition of local differential privacy, and overview the truth inference algorithms. The notations used in the paper are shown in Table \ref{tb:notation}. 
\begin{table}[!hbpt]
\begin{small}
\centering
\begin{tabular}{|c|c|}
\hline
Symbol & Meaning \\\hline
$m$ & \# of workers \\\hline
$n$ & \# of tasks \\\hline
$\vec{a}_i$ & Worker $W_i$'s original answer vector \\\hline
$\hat{a}_j$ & Worker $W_i$'s answer vector after perturbation \\\hline
$a_{i,j}$ & Worker $W_i$'s answer to task $T_j$ \\\hline
$\Gamma$ & Domain of task answers (excluding NULL answers) \\\hline
$\sigma_i$ & Standard deviation of worker $W_i$'s answer error\\\hline
$q_i$ & Estimated quality of worker $W_i$ \\\hline
$\mu_j$/$\hat{\mu_{j}}$ & Real/estimated truth of task $T_j$ \\\hline
$s_i$ & Percentage of tasks that $W_i$ returns non-NULL values \\\hline
$\epsilon$ & Privacy budget \\\hline
$\mathcal{T}_i$ & The set of tasks that worker $W_i$ performs\\\hline
$d$ & Factorization parameter\\\hline
\end{tabular}
\end{small}
\caption{\label{tb:notation} Notations}
\vspace{-0.2in}
\end{table}

\vspace{-0.1in}
\subsection{Local Differential Privacy (LDP)}
\label{sc:ldp}
\vspace{-0.1in}

The need for data privacy arises in two different contexts: the {\em local
privacy} context in which the individuals disclose their personal information (e.g., participation in surveys for specific population), and the {\em global privacy} context in which the institutions release aggregation 
information of their data (e.g., US Government releases
census data). 
The classic concept of differential privacy (DP) \cite{dwork2008differential} is proposed in the global privacy context. 
In a nutshell, DP ensures that an adversary cannot infer whether or not a particular individual is participating in the database query, even with unbounded computational power and access to every entry in the database except for that particular individual's data. DP considers a centralized setting that includes a trusted data curator, who generates the perturbed statistical information (e.g., counts and histograms) by using a randomized mechanism \cite{hay2016principled,li2014dpsynthesizer,Mohammed:2011:DPD:2020408.2020487,xiao2011differential}. 
Recently, a variant of DP, named {\em local differential privacy} (LDP) \cite{kasiviswanathan2011can,duchi2013local} was proposed for the decentralized setting where multiple data providers send their private data to a untrusted data curator. 
Before sending it to the data curator, each data provider perturbs his private data locally by using a differentially private mechanism. 
The randomized response method \cite{warner1965randomized} has been used to provide local privacy when individuals respond to sensitive surveys. Intuitively, for a given survey question, respondents provide 
a truthful answer with probability $p > 1/2$ and lie with probability $1-p$. 
The randomized response method provides $(ln\frac{p}{1-p})$-LDP \cite{wang2017locally}. 

While LDP considers {\em per-user privacy}, in this paper, we consider {\em per-user per-answer privacy}. Therefore, by following \cite{qin2017generating}, we define a variant of LDP  \cite{duchi2013local}, named {\em $\epsilon$-cell local differential privacy}, below. 

\begin{defi}[$\epsilon$-cell Local Differential Privacy]
\label{df:neighbor}
A randomized privatization mechanism $\mathcal{M}$ satisfies $\epsilon$-cell local differential privacy ($\epsilon$-cell LDP) iff for any pair of answer vectors $\vec{a}$ and $\vec{a}'$ that differ at one cell, we have:
\[\forall \vec{z}_p\in Range(\mathcal{M}): \frac{Pr[\mathcal{M}(\vec{a})=\vec{z}_p]}{Pr[\mathcal{M}(\vec{a}')=\vec{z}_p]} \leq e^{\epsilon},\]
where $Range(\mathcal{M})$ denotes the set of all possible outputs of the algorithm $\mathcal{M}$. 
\end{defi}

\nop{Formally, given a randomized algorithm $\mathcal{M}$, and let $S$ be an arbitrary set of output of $\mathcal{M}$, then $\mathcal{M}$ satisfies $\epsilon$-{\em differential privacy} if for any two neighboring datasets $D$ and $D'$ that differ at one tuple, 
\begin{equation}
\label{eq:dp}
P[\mathcal{M}(D)\in S] \leq e^{\epsilon}P[\mathcal{M}(D')\in S]
\end{equation}
Intuitively, DP limits the attacker's ability to infer the original dataset from the output of $\mathcal{M}$.
An effective approach to achieve DP is to add Laplace noise to the original computation result, i.e., 
\begin{equation}
\label{eq:dp_lap}
\mathcal{M}(D)=f(D)+Lap(\frac{\Delta f}{\epsilon})
\end{equation}
where $f$ is a function on the dataset, $Lap(\cdot)$ is the Laplace distribution, abd $\Delta f$ is the {\em sensitivity} of $f$, which measures the maximum deviation of $f$ on two input datasets that differ at only one element. 
}
\nop{
Most of the existing DP and LDP models do not give special treatment to the NULL values. Therefore, next, we will present the formal privacy definition of LDP with the presence of NUll values. We first introduce the concept of {\em neighboring databases}.  
Given two databases $D$ and $D'$, we say they are ``adjacent'' (or ``neighboring'') if they differ at only one element.
Assume they differ at the $k$-th element (denoted as $D_k$ and $D'_k$). 
Considering the potential existence of NULL values, there are two possible types of neighboring databases: 
\begin{itemize}
	\item Neither of $D_k$ and $D'_k$ is NULL, and $D_k\neq D'_k$; 
    \item $D_k$ is NULL, but $D'_k$ is not, and vice versa. 
\end{itemize}
}

\nop{
Given two answer vectors $\vec{a_i}$ and $\vec{a_j}$ for $n$ tasks, $\vec{a_i}$ and $\vec{a_j}$ are neighbors if they differ at one element, i.e., for some $1\leq k\leq n$, $a_{i,k}\neq a_{j,k}$. Considering the {\em NULL} values in answer vectors, it is possible that:
\begin{itemize}
	\item Neither of $a_{i,k}$ and $a_{j,k}$ is NULL, but $a_{i,k}\neq a_{j,k}$; or
    \item $a_{i,k}$ is NULL, but $a_{j,k}$ is not; and vice versa. 
\end{itemize}
It is worth noting that if $\vec{a_i}$ only includes the answer for a single task (and NULL values for the other tasks), we do not consider the answer vector of $\vec{a_j}$ with all NULL values as an neighboring vector, as it violates our assumption in Section \ref{sc:framework}.
}


\vspace{-0.1in}
\subsection{Truth Inference}
\label{section:truth_inference}
\vspace{-0.1in}

The goal of truth inference is to infer the true answer (i.e., truth) by integrating the noisy answers from the workers. As the workers may produce answers of different quality, most of the existing truth inference algorithms take the
quality of workers into consideration. Intuitively, the answers by workers of higher quality are more likely to be considered as the truth. 
The challenge is that worker quality is usually unknown a priori in practice. To tackle this challenge, quite a few truth inference algorithms (e.g. \cite{guo2012so,liu2012cdas,zhang2013reducing,cao2012whom}) have been designed to infer both worker quality and the true answers of the tasks. Intuitively, worker quality and the inferred  
truth are correlated: workers whose answers are closer to true answers more often will be assigned higher quality, and answers that are provided by high-quality workers will have higher influence on the truth. 

In this paper, we follow the truth inference algorithm \cite{li2014confidence, zhao2012probabilistic} to estimate the true answers (truth) from the workers' answers. 
Formally, given a set of workers $\mathcal{W}$, and a set of tasks $\mathcal{T}$, for any worker $W_i\in \mathcal{W}$, let $\mathcal{T}_i\subseteq \mathcal{T}$ be the set of tasks that worker $W_i$ performs. 
For any task $T_j\in\mathcal{T}_i$, let $\overline{W_{j}}$ be the set of workers who perform it, $a_{i,j}$ be the answer that $W_i$ provides to $T_j$, $\hat{\mu_{j}}$ be the estimated truth of $T_j$, and $q_i$ be the quality of worker $W_i\in \mathcal{W}$. 
We follow the same assumption as that in \cite{li2014confidence, zhao2012probabilistic} that the error of worker $W_i$ follows the normal distribution $\mathcal{N}(0, \sigma_i^2)$, where $\sigma_i$ is the standard error deviation of $W_i$. Intuitively, the lower $\sigma_i$ is, the higher the worker quality $q_i$ will be. 
We also follow the assumption in most of the truth inference works (e.g., \cite{li2014confidence, zhao2012probabilistic}) that the worker quality stays stable across all the tasks. 

The truth inference algorithms \cite{li2014confidence, zhao2012probabilistic} 
follow an iterative style. 
Algorithm \ref{alg:estimate} presents the pseudo-code. 
Initially, each worker is assigned the same quality $\frac{1}{m}$. Then the weighted average of the worker answers are computed as the estimated truth. Specifically, the estimated truth $\hat{\mu}_{j}$ of task $T_j$ is computed as: 
\begin{equation}
\label{td_truth}
\hat{\mu_{j}}=\frac{\sum_{W_i\in \overline{W_{j}}}{q_i\times a_{i,j}}}{\sum_{W_i\in \overline{W_{j}}} q_i}
\end{equation}
Based on the estimated truth $\hat{\mu_{j}}$, the worker quality is updated accordingly. Intuitively, if a worker provides accurate answers more often, he/she has a better quality. Specifically, the quality $q_i$ of worker $W_i$ is inversely related to the total difference between his answers and the estimated truth. We adopt the weight estimation method \cite{li2014confidence} for worker quality. Namely, the quality $q_i$ of worker $W_i$ is computed as: 
\begin{equation}
\label{td_weight}
  q_i\propto \frac{1}{\sigma_i}=\frac{1}{\sqrt{\frac{1}{|\mathcal{T}_i|}\sum_{t_j\in \mathcal{T}_i} {(a_{i,j}-\hat\mu_{j})^2}}}
\end{equation}
The estimated truth $\{\hat{\mu_j}\}$ and worker quality $\{q_i\}$ are kept updated iteratively until they reach convergence. 
\vspace{-0.1in}
\begin{algorithm}[!hbtp]
  \caption{Truth inference}
  \label{alg:estimate}
  \begin{algorithmic}[1]
    \REQUIRE The workers' answers $\{a_{i,j}\}$
    \ENSURE The estimated true answer (i.e., the truth) of tasks $\{\hat{\mu_{j}}\}$ and the quality of workers $\{q_i\}$ 
    \STATE Initialize worker quality $q_i=1/m$ for each worker $W_i\in\mathcal{W}$;
    \WHILE{the convergence condition is not met}
        \STATE Estimate $\{\hat{\mu_{j}}\}$ following Equation (\ref{td_truth});
        \STATE Estimate $\{q_i\}$ following Equation (\ref{td_weight});
    \ENDWHILE
  \RETURN $\{\hat{\mu_{j}}\}$ and $\{q_i\}$;
  \end{algorithmic}
\end{algorithm}
\vspace{-0.2in}

To evaluate data utility, we use the {\em mean absolute error} (MAE) to measure the accuracy of the inferred truth. Specifically, 
\vspace{-0.1in}
\begin{equation}
\label{eq:td_error}
MAE = \frac{\sum_{T_j\in \mathcal{T}} |\mu_j-\hat{\mu_j}|}{n},
\end{equation}
where $\mu_j$ ($\hat{\mu_j}$, resp.) is the real (estimated, resp.) true answer of task $T_j$. 

\nop{
Next, we present the error bound of the truth inference algorithm.
\Wendy{why do we need to discuss the error bound of the truth inference algorithm (without DP)? Is it because you will compare with the truth inference algorithm with DP later?}
\Boxiang{The proof of Theorem 2.1 provides a general framework for error bound inference based on truth discovery. In every DP mechanism, to prove its error bound, we just need to replace the part $\mid \mu_j-a_{ij}\mid$ in Line 6 of Equation (14) in Appendix A. Since we put the proof in Appendix, we can remove Theorem 2.1}
}


\vspace{-0.15in}
\subsection{Matrix Factorization}
\label{sc:matrix}
\vspace{-0.1in}
One of the popularly used methods to handle missing values is matrix factorization \cite{koren2009matrix}. 
Formally, given a $m\times n$ matrix $M$ with NULL values, matrix factorization finds two profile matrices $U$ and $V$, where $U$ is an $m\times d$ matrix and $V$ is a $d\times n$ matrix, such that $M\approx UV$. 
The value $d$ is normally chosen to be smaller than $m$ or $n$. 
We call $d$ the {\em factorization parameter}. 
We define the loss function \cite{friedman2016differential} of matrix factorization: 
\begin{equation}
\label{eq:loss}
L(M, U, V) = \sum_{(i,j)\in \Omega} (M_{i,j}-\vec{u}_i^T\vec{v}_j)^2
\end{equation}
where $\Omega$ is the set of observed answers in $M$.
To ensure the accuracy of the estimated $U$ and $V$, an effective method is to follow {\em stochastic gradient descent} (SGD) to learn two latent matrices $U$ and $V$ that minimize the loss function: 
 \begin{equation}
\label{eq:uv}
(U,V)=\argmin_{U,V}L(M,U,V)
\end{equation}
By matrix factorization, each worker $W_i$ is characterized by a {\em worker profile vector} $\vec{u}_i$, and each task $T_j$ is characterized by a {\em task profile vector} $\vec{v}_j$. The worker $W_i$'s answer of task $T_j$, which is denoted as $M_{i,j}$, can be approximated by the inner product of $\vec{u}_i$ and $\vec{v}_j$, i.e., $\vec{u}_i^T \vec{v}_j$. 
\vspace{-0.1in}
\section{Preliminaries}
\label{sc:preliminaries}
\vspace{-0.1in}
\subsection{Crowdsourcing Framework}
\label{sc:framework}
\vspace{-0.1in}
In general, a crowdsourcing framework consists of two parties: (1) 
the {\em task requester} (TR) who generates the tasks and releases them on a crowdsourcing platform  (e.g., Amazon Mechanical Turk \cite{amt}); 
(2) the {\em workers} who provide their answers to the tasks via the crowdsourcing platform. 
For the rest of the paper, we use the terms {\em task requester (TR)} and {\em data curator (DC)} interchangeably.

Consider $n$ tasks $\mathcal{T}=\{T_1, \dots, T_n\}$ and $m$ workers $\mathcal{W}=\{W_1, \dots, W_m\}$. Worker $W_i$'s answers are represented as a row vector $\vec{a}_i$, in which the element $a_{i,j}$  denotes the answer of worker $W_i$ to the task $T_j$. 
We assume that all the answers are in numerical format (e.g., image classification categories and product ratings). The domain of non-NULL answers is denoted as $\Gamma$.  
Each worker can access all the tasks in $\mathcal{T}$ (as on Amazon Mechanical Turk), and choose a subset of task(s) of $\mathcal{T}$ to provide answers. We denote $a_{i,j}=NULL$ if worker $W_i$ does not provide any answer to task $T_j$. 
We assume that each worker performs at least one task (i.e., each answer vector has at least one non-NULL value).

\vspace{-0.1in}
\subsection{Problem Definition}
\vspace{-0.1in}
We assume that DC is untrusted. Therefore, releasing the original worker answers to DC may reveal sensitive information. 
Obviously, non-NULL answers must be protected, since they reveal the workers' sensitive information (e.g., political opinions and medical conditions). On the other hand, the NULL values reveal the fact  whether a worker participates in a task or not. Such fact can be sensitive. For instance, in medical research, a patient's refusal to answer some questions could reveal that the patient may have some specific diseases and thus is not qualified to answer these questions \cite{ciglic2014k}. Therefore, NULL values must be protected too. 
Our goal is to design privacy-preserving algorithms to protect both non-NULL answers and NULL values, while enable truth inference with high accuracy on the collected worker answers. We formalize the problem statement below. 

\noindent{\bf Problem statement.} Given a set of workers $\mathcal{W}$, their answer vectors 
$A = \{\vec{a}_i\}$, and a non-negative privacy parameter $\epsilon$, we design a randomized privatization mechanism $\mathcal{M}$ such that for each worker $W_i\in\mathcal{W}$ and his/her answer vector $\vec{a}_i$, $\mathcal{M}$ provides $\epsilon$-cell LDP on $a_{i,j}$, for each $a_{i,j}$ of $\vec{a}_i$. We prefer to minimize MAE of the truth inferred from $A^P =\{\mathcal{M}(\vec{a}_i)|\forall \vec{a}_i\in A\}$.

\vspace{-0.1in}
\subsection{Our Solutions in a Nutshell}
\vspace{-0.1in}
We first propose easy extension of two existing LDP approaches, namely the {\em Laplace perturbation} ({\bf LP}) method which applies perturbation on worker answers directly, and the {\em randomized response} ({\bf RR}) method that generates answers by following a randomized way. 
Our theoretical analysis shows that these easy extension of both LP and RR approaches have large error bound of truth inference, especially when the data is very sparse. 
Therefore, we design a new matrix factorization ({\bf MF}) privatization method to deal with sparse worker answers under LDP. We observe that the adaption of a worker profile vector only relies on the worker's own answers rather than those of other workers. Therefore, the matrix factorization can be performed locally. However, applying LDP to matrix factorization is not trivial. 
Some existing works of differentially private matrix factorization \cite{balu2016differentially,friedman2016differential} assume DC is trusted; they do not fit the LDP model. \cite{hua2015differentially,shin2018privacy} apply LDP on matrix factorization. They follow the same strategy: each worker generates his worker profile matrix locally, and learns the task profile matrix by iterative communication with DC, which can incur substantial communication overhead. 
To address this issue, we design a new protocol by which: (1)   DC generates the task profile matrix initially (independent from the worker answers), and sends it to the workers; and (2) each worker learns the worker profile vector from his local data only.  This eliminates the communication between DC and workers. LDP is achieved by adding perturbation on the objective function of matrix factorization. We are aware that a task profile matrix generated independently from the worker answers is not as accurate as that learned from the worker answers. There exists the trade-off between the overhead of the LDP approach and the accuracy of truth inference. Both of our theoretical and empirical analysis show that, even with a task profile matrix that is independent from worker answers, the error of truth inference results on the perturbed worker answers is small.


\nop{
To overcome the drawbacks, we let each worker factorizes his answer vector locally with a fixed task factor matrix $V$, which is specified by the DC. 
We adapt the objective perturbation method \cite{chaudhuri2009privacy} to further reduce the noise. The key idea of objective perturbation is to achieve DP by randomly perturbing the objective function instead of perturbing the output of the algorithm. 
\Ting{It is not clear how to perturb the factorazation results directly, and why it will incur large noises.}\Boxiang{Like LP, we can add Laplace noise to each factorized answer directly. If we do that, there are two sources of noise. The inaccuracy of missing value imputation, and the DP noise.} \Ting{ By matrix factorization results do you mean the results of U*V? or the answer vector times V? If so, each element of the output would depend on all answers from a worker, which means we cannot simply add noise independently to each element. If I understand correctly, then more explanation should be provided. This is to show this is not a trivial problem. }
\Boxiang{Actually, the MF result is $\vec{u}_i V$, where $\vec{u}_i$ is the worker profile vector for $W_i$. We cannot let the workers exchange their answers, so we let each worker factorize his answers independently.}
}

\vspace{-0.1in}
\section{Extension of Existing LDP Approaches}
\label{sc:straw}
\vspace{-0.1in}

In this section, we present the easy extension of two existing LDP approaches, namely Laplace perturbation (LP) and randomized response (RR), to deal with data sparsity. 
	\vspace{-0.1in}
\subsection{Laplace Perturbation (LP)}
\label{sc:lp}
\vspace{-0.1in}
Laplace perturbation is one of the most well-known approaches to achieve DP. Typically Laplace perturbation is applied on aggregation results (e.g., count and sum) in the centralized setting. It can be easily extended to provide LDP by adding Laplace noise to each worker answer independently. To handle NULL values, a straightforward solution is to replace NULL values with some non-NULL value in the answer domain $\Gamma$. 
\nop{
We investigate two strategies: 
\begin{itemize}
\item {\em Fixed-value replacement}: All NULL values are replaced with the same constant (e.g., the smallest value in $\Gamma$); 
\item {\em Random-value replacement}: Each NULL value is replaced with a value randomly picked from $\Gamma$ by following a probability distribution (e.g., uniform distribution). 
\end{itemize}
}
Formally, the NULL-value replacing strategy can be defined as a conversion function $g(\cdot)$ as follows:
\begin{equation}
\label{eq:g}
g(a_{i,j})=\left\{
\begin{aligned}
v &\vspace{4cm}& a_{i,j}=NULL\\
a_{i,j} &\vspace{4cm}& a_{i,j}\not=NULL,
\end{aligned}
\right.
\end{equation}
where $v$ is an non-NULL answer in $\Gamma$.  
After conversion, Laplace noise is added to each element in the answer vector $\vec{a}_i$. Formally, 
\resizebox{0.43\textwidth}{!}{
\begin{minipage}{\linewidth}
\begin{equation}
\mathcal{L}(\vec{a}_i)=\big(g(a_{i,1})+Lap(\frac{|\Gamma|}{\epsilon}),g(a_{i,2})+Lap(\frac{|\Gamma|}{\epsilon}),...,g(a_{i,n})+Lap(\frac{|\Gamma|}{\epsilon})\big),
\end{equation}
\end{minipage}
}

The following theorem shows the privacy guarantee of the easy extension of the LP approach. 
\begin{thm}
\label{th:lp}
The LP mechanism guarantees $\epsilon$-cell LDP for both cases that $g(\cdot)$ replaces NULL with a constant value in $\Gamma$ and a random value picked by following any arbitrary distribution over $\Gamma$.
\end{thm}
Due to the space limit, we defer the proof to our full paper \cite{sun2018truth}.
One of the weakness of the LP approach is that it replaces a large number of NULL values and may change the original data distribution significantly. Next, we show the error bound of the inferred truth from the answers with LP perturbation. 
\begin{thm}
\label{thm:lp_bound}
Given a set of answer vectors $A = \{\vec{a}_i\}$, let $A^P = \{\hat{a}_i\}$ be the answer vectors after applying LP on $A$. Then the expected error $E\left[MAE(A^P)\right]$ of the estimated truth on $A^P$ must satisfy that 
 \[E\left[MAE(A^P)\right] \leq \frac{1}{n}\sum_{j=1}^{n}\sum_{i=1}^{m}(q_i\times e_{i,j}^{LP}),\] 
where $e_{i,j}^{LP}=(1-s_i)\left(\phi_{j}+\frac{|\Gamma|}{\epsilon}\right)+s_i\left(\sigma_i\sqrt{\frac{2}{\pi}}+\frac{|\Gamma|}{\epsilon}\right)$, 
$\mu_j$ is the ground truth of task $T_j$, $\sigma_i$ is the standard error deviation of worker $W_i$, $s_i$ is the fraction of the tasks that $W_i$ returns non-NULL values, and $\phi_{j}$ is the deviation between $\mu_j$ and the expected value $E(v)$ of $v$ in Equation (\ref{eq:g}). 
\end{thm} 
The expected valued $E(v)$ is decided by how $v$ is sampled in Equation (\ref{eq:g}). For example, if $v$ is sampled by a uniform distribution, then $E(v)=\frac{min+max}{2}$, where $min$ and $max$ are the minimum and maximum values of $\Gamma$. 
We present the proof of Theorem \ref{thm:lp_bound} in our full paper \cite{sun2018truth}. 
Intuitively, Theorem \ref{thm:lp_bound} shows that the sparsity (i.e., $1 - s_i$) affects the error bound. 
In particular, for sparse answers, LP method can incur significant inaccuracy to the estimated truth in theory. For example, consider a simple scenario where all the workers have the same quality, i.e., $q_i=\frac{1}{m}$ and $\sigma_i=1$, the truths of all tasks are $0$, $s_i=0.1$ for all workers, $|\Gamma|=10$, and $\epsilon=1$. We have $E\left[MAE(\{\hat{a}_j\})\right]\leq 14.13$. 
Our experimental results (Section \ref{sc:exp}) will also show that LP leads to high MAE of truth inference on the sparse data. 

	\vspace{-0.1in}
\subsection{Randomized Response (RR)}
\label{sc:rr}
\vspace{-0.1in}
An alternative to realizing differential privacy is  {\em randomized response} (RR), first proposed in  \cite{warner1965randomized}. Intuitively, the private input is perturbed with some known probability. 
However, none of the existing randomized response solutions \cite{wang2016using,dwork2006calibrating,erlingsson2014rappor,karwa2014differentially} has defined the probability for NULL values. 
\nop{
Formally, the probability is specified as follows \cite{wang2016using}: 
 \begin{equation}
	\label{eq:rr}
	\forall y\in D,\ Pr[\mathcal{M}(v) = y]=
	\begin{cases}
	\frac{e^{\epsilon}}{|\Gamma|-1 + e^{\epsilon}} & \text{ if } y = v\\
	\frac{1}{|\Gamma| - 1 + e^{\epsilon}} & \text{ if 		}y \neq v
	\end{cases}
	\end{equation}   
}   
We extend the randomized response approach to deal with NULL values. The key idea is that, given the domain of non-NULL answers $\Gamma$,  we add the NULL value as an answer to $\Gamma$. Then given an answer vector $\vec{a}_i$, for each element $a_{i,j}\in\vec{a}_i$, 
\vspace{-0.1in}
 \begin{equation}
	\label{eq:rr}
	\forall y\in\Gamma,\ Pr[\mathcal{M}(a_{i,j}) = y]=
	\begin{cases}
	\frac{e^{\epsilon}}{|\Gamma| + e^{\epsilon}} & \text{ if } y = a_{i,j}\\
	\frac{1}{|\Gamma| + e^{\epsilon}} & \text{ if 		}y \neq a_{i,j}
	\end{cases}
	\end{equation}   
Intuitively, each original worker answer either remains unchanged in the perturbed answer vector with probability $\frac{e^{\epsilon}}{|\Gamma|+e^{\epsilon}}$ or is replaced with a different value with probability $\frac{1}{|\Gamma|+e^{\epsilon}}$. 
Note that a NULL value can be replaced with a non-NULL value, and vice versa. 
    
Obviously, the RR approach satisfies $\epsilon$-cell LDP. The proof is similar to that the original randomized response approach can provide LDP \cite{wang2016using}. 
\nop{
\begin{thm}
\label{thm:rr_dp}
The RR mechanism guarantees $\epsilon$-LDP.

\nop{
\begin{proof}
To prove that the RR mechanism satisfies $\epsilon$-LDP on each answer, it suffices to show that for any pair of answer vectors $\vec{a}_i$ and $\vec{a}_j$ that differ at one element (i.e., the answer of a specific task), $\frac{Pr[\mathcal{M}(\vec{a}_i)=\vec{z}_p]}{Pr[\mathcal{M}(\vec{a}_j)=\vec{z}_p]}\leq e^{\epsilon}$.
\Boxiang{I realize that we do not have to prove $\frac{Pr[\mathcal{M}(\vec{a}_i)=\vec{z}_p]}{Pr[\mathcal{M}(\vec{a}_j)=\vec{z}_p]}\leq e^{\epsilon}$. We just need to prove $\frac{Pr[\mathcal{M}(v)=y]}{Pr[\mathcal{M}(v')=y]} \leq e^{\epsilon}$. If we do so, the proof is going to be very straightforward.}\Wendy{I leave this decision to you if the proof needs to be changed.}
\Boxiang{I am in a dilemma now. If we change the proof, the proof of Theorem 4.3 would be too simple. I am afraid it may make the reviewer think our paper does not have sufficient depth. If we do not change it, the reviewer may think we prove a simple thing in a relatively complicated way. If the space is sufficient, maybe we can change the proof and bring the proof of Theorem 4.1 back to the paper.}
Without loss of generality, we assume that $\vec{a}_i$ and $\vec{a}_j$ only differ at the first element. Therefore, we have:
\begin{equation*}
\begin{aligned}
\frac{P[\mathcal{M}(\vec{a}_i)=\vec{z}_p]}{P[\mathcal{M}(\vec{a}_j)=\vec{z}_p]} & = \frac{\Pi_{k=1}^n Pr[\mathcal{M}(a_{i,k})=z_{pk}]}{\Pi_{k=1}^n Pr[\mathcal{M}(a_{j,k})=z_{pk}]} \\
& = \frac{Pr[\mathcal{M}(a_{i,1})=z_{p1}]}{Pr[\mathcal{M}(a_{j,1})=z_{p1}]}\\
& \leq max_{z_{p1}} \frac{Pr[\mathcal{M}(a_{i,1})=z_{p1}]}{Pr[\mathcal{M}(a_{j,1})=z_{p1}]} \\
& = \frac{\frac{e^{\epsilon}}{|\Gamma| + e^{\epsilon}}}{\frac{1}{|\Gamma| + e^{\epsilon}}}\\
& = e^{\epsilon}
\end{aligned}
\end{equation*}

Hence, we conclude that the RR mechanism satisfies $\epsilon$-LDP.
\end{proof}   
}
\end{thm} 

The proof of Theorem \ref{thm:rr_dp} is similar to the proof that the randomized response can provide DP \cite{wang2016using}. 
}
Next, we show the expected error bound of the RR approach by the following theorem.
\begin{thm}
\label{thm:rr_bound}
Given a set of answer vectors $A = \{\vec{a}_i\}$, let $A^P = \{\hat{a}_i\}$ be the answer vectors after applying RR on $A$. Then the expected error $E\left[MAE(A^P)\right]$ of the estimated truth on $A^P$ must satisfy that
\vspace{-0.1in}
\[E\left[MAE(A^P)\right] \leq \frac{1}{n}\sum_{j=1}^{n}\frac{\sum_{W_i\in \overline{W_{j}}}q_i\times e_{i,j}^{RR}}{\sum_{W_i\in \overline{W_{j}}}q_i},\] 
where
\begin{equation*}
\begin{aligned}
e_{i,j}^{RR}=&(1-s_i)\left|\mu_j-\sum_{y\in\Gamma}y\frac{1}{e^{\epsilon}
+|\Gamma|}\right|\\
+&\sum_{x\in\Gamma}s_i\mathcal{N}(x;\mu_j,\sigma_i)\left|\mu_j-\sum_{y\in\Gamma}yP_{xy}\right|,
\end{aligned}
\end{equation*}
where $s_i$ is the fraction of tasks that worker $W_i$ returns non-NULL values, $P_{xy}$ is the probability that value $x$ is replaced with $y$, and $\mathcal{N}(x;\mu_j,\sigma_i^2)=\frac{1}{\sqrt{2\pi \sigma_i^2}}\exp\left(-\frac{(x-\mu_j)^2}{2\sigma_i^2}\right)$ is the probability to pick the answer $x$ from $\Gamma$ by following the normal distribution $\mathcal{N}(\mu_j, \sigma_i^2)$. 
\end{thm}
We include the proof of Theorem \ref{thm:rr_bound} in our full paper \cite{sun2018truth}. 
In general, RR method has a large error bound of the estimate truth, especially when the data is sparse. For example, consider the setting where $q_i=\frac{1}{m}$, $\sigma_i=1$, $\Gamma=[0,9]$, $\mu_j=0$, and $\epsilon=1$.
When $s_i=0.1$ (i.e., a very sparse answer vector), $E\left[MAE(\{\hat{a}_j\})\right]\leq 3.551$. This is large, considering the fact that the domain size is 10. 
\nop{The error bound is influenced by the sparsity. Consider the same setting. When $s_i=0.9$ (i.e., a dense answer vector), $E\left[error(\{\hat{a}_j\})\right]$ is increased to 3.639. \Wendy{why the error bound increases when the sparsity decreases? This is counter-intuitive. Note that this is different from the observation from the LP approach.}
\Boxiang{I double-checked it. The formula can be simplified to $(1-s_i)3.54+3.65s_i$. So according to the bound, it is expected that it increases with $s_i$. But this is actually from the intuition and the observation in Figure 3(b). I guess this is because we use normal distribution to model the worker answers, but they are actually discrete values. So the expecation of them are not identical.}}

\nop{
\begin{theorem}
\label{thm:rr_bound}
The expected error of estimated truth $\{\hat{a}_j\}$ from algorithm \ref{alg:estimate} will not exceed $\frac{1}{n}\sum_{j=1}^{n}\frac{\sum_{W_i\in \overline{W_{j}}}q_i\times e_{i,j}^{RR}}{\sum_{W_i\in \overline{W_{j}}}q_i}$.
\end{theorem}
\vspace{0.1in}
\begin{proof}
Based on the proof in equation (\ref{equ:error_original}), we just need to calculate $E\left[\left|\mu_j-a_{i,j}\right|\right]$ in its line 6:
\begin{equation}
\begin{aligned}
&E\left[\left|\mu_j-a_{i,j}\right|\right]\\
&=\frac{1}{s_i+e^{\epsilon_1}(1-s_i)}\left(s_i E\left[\left|\mu_j-\mathcal{U}(\mathcal{V})\right|\right]+e^{\epsilon_1}(1-s_i)E\left[\left|\mu_j-\mathcal{P}_{\epsilon_2,\mathcal{V}}(\cdot|\mathcal{N}(\mu_j,\sigma_i))\right|\right]\right)\\
&\le\frac{1}{s_i+e^{\epsilon_1}(1-s_i)}\left(s_i (\mu_j+E[|\mathcal{U}(\mathcal{V})|])+e^{\epsilon_1}(1-s_i)(\mu_j+E[|\mathcal{P}_{\epsilon_2,\mathcal{V}}(\cdot|\mathcal{N}(\mu_j,\sigma_i))|])\right)\\
&=\frac{1}{s_i+e^{\epsilon_1}(1-s_i)}\left(s_i (\mu_j+\psi(\mathcal{V}))+e^{\epsilon_1}(1-s_i)(\mu_j+\phi_{\sigma_i,\mathcal{V}}(\mu_j,\epsilon_2)|])\right)\\
&=e_{i,j}^{RR}
\end{aligned}
\end{equation}
where $\psi(\mathcal{V})$ is the expectation of uniform distribution on domain of answers $\mathcal{U}(\mathcal{V})$ that knowing the cell is a null as a prior:
\begin{equation}
\phi(\mathcal{V})=E[|\mathcal{U}(\mathcal{V})|]=\frac{1}{|\mathcal{V}|}\sum_{x\in\mathcal{V}}x
\end{equation}
$\psi(\cdot)$ is the distribution of values when knowing the cell is originally a value as a prior:

\begin{equation}
\begin{aligned}
&\phi_{\sigma_i,\mathcal{V}}(\mu_j,\epsilon_2)=E[|\mathcal{P}_{\epsilon_2,\mathcal{V}}(\cdot|\mathcal{N}(\mu_j,\sigma_i))|]\\
&=\sum_{x\in\mathcal{V}}\left[\frac{\mathcal{N}(x;\mu_j,\sigma_i)+\mathcal{N}(-x;\mu_j,\sigma_i)}{e^{\epsilon_2}+|\mathcal{V}|-1}\sum_{y\in\mathcal{V}}\left[ye^{\epsilon_2}\cdot \mathds{1}(x=y)+y\cdot \mathds{1}(x\not=y)\right]\right]
\end{aligned}
\end{equation}
\end{proof}
}

\nop{
    In order to guarantee $\epsilon$-DP and minimize the distortion between the private value and the perturbed value, i.e., to maximize the probability that the perturbed value remains the same as the original value, it is desired that

To provide $\epsilon$-cell LDP, the client cannot simply apply the techniques in \cite{wang2016using} due to the existence of NULL values. 
A simple remedy is to take NULL as an auxiliary value and construct a $(s+1)\times (s+1)$ design matrix following Equation (\ref{eq:uv}). 
However, it would degrade the accuracy in the estimated truth of the crowdsourced tasks dramatically, especially if the domain size $s$ is large. This is because $P[y=NULL|x=NULL]=\frac{e^{\epsilon}}{s-1+e^{\epsilon}}$, which is inversely proportional to $s$. If $s$ is large, the client replaces a majority of NULL values with a random real answer. 

To address this issue, we propose a two-step randomized response (RR) mechanism that avoids random imputation of the NULL values in the worker answers. 
The key idea is to create a generic class NON-NULL for any value that is not NULL. Given any private answer vector $\vec{a_i}$, for each cell  $a_{ij}$ ($1\leq j\leq n$), in the first step, the client randomly chooses to replace it with a NULL or NON-NULL value. 
If the choice is NON-NULL, in the second step, the client randomly picks a value from the answer domain $\Gamma$. 

\noindent{\bf Step 1. Choosing a generic class.} Consider the private answer $a_{ij}$, which could be any value in $\Gamma$ or NULL, the client decides to replace it with NULL or NON-NULL by the following design matrix:
\begin{equation}
P^A = 
\kbordermatrix{
    & NULL & NON-NULL \\
    NULL & \frac{e^{\epsilon_1}}{1+e^{\epsilon_1}} & \frac{1}{1+e^{\epsilon_1}} \\
    NON-NULL & \frac{1}{1+e^{\epsilon_1}} & \frac{e^{\epsilon_1}}{1+e^{\epsilon_1}}
    },
\end{equation}
where $\epsilon_1>0$ is the privacy budget that adjusts the random response between NULL and non-NULL answers.
Following this design, the NULL values are more likely to be kept NULL in the perturbed answer vector. 

\noindent{\bf Step 2. Choosing a real value.} If the client picks the NON-NULL class in the first step, he picks a value from the domain $\Gamma$ based on the design matrix in Equation (\ref{eq:p2}). Without loss of generality, we assume that the domain is $\{1, \dots, s\}$.
\begin{equation}
\label{eq:p2}
P^B = 
\kbordermatrix{
& 1 & 2 & \hdots & \hdots & s \\
NULL & \frac{1}{s} & \frac{1}{s} & \hdots & \hdots & \frac{1}{s} \\
1 & \frac{e^{\epsilon_2}}{s-1+e^{\epsilon_2}} & \frac{1}{s-1+e^{\epsilon_2}} & \hdots & \hdots & \frac{1}{s-1+e^{\epsilon_2}} \\
2 & \frac{1}{s-1+e^{\epsilon_2}} & \frac{e^{\epsilon_2}}{s-1+e^{\epsilon_2}} & \frac{1}{s-1+e^{\epsilon_2}} & \hdots & \frac{1}{s-1+e^{\epsilon_2}} \\
\vdots & \vdots & \vdots & \vdots & \ddots & \frac{1}{s-1+e^{\epsilon_2}}\\
s & \frac{1}{s-1+e^{\epsilon_2}} & \frac{1}{s-1+e^{\epsilon_2}} & \hdots & \hdots & \frac{e^{\epsilon_2}}{s-1+e^{\epsilon_2}} 
},
\end{equation}
where $\epsilon_2>0$ is the privacy budget for choosing an answer from $\Gamma$. 
It is worth noting that $P^B$ is a $(s+1)\times s$ design matrix which provides the mapping from any (NULL and NON-NULL) value to the answer domain $\Gamma$.
In consistency with Equation (\ref{eq:uv}), we maximize the probability of $P^B_{uv}$ if $u=v$. 

To be clear, by combining $P^A$ and $P^B$, the design matrix in the two-step RR mechanism is 
\begin{equation}
\label{eq:uv2}
P_{uv} = \left\{
\begin{array}{lr}
\frac{e^{\epsilon_1}}{1+e^{\epsilon_1}} & \text{ if } u=NULL \text{ and } v=NULL \\
\frac{1}{s}\frac{1}{1+e^{\epsilon_1}} & \text{ if } u=NULL \text{ and } v\neq NULL \\
\frac{1}{1+e^{\epsilon_1}} & \text{ if } u\neq NULL \text{ and } v=NULL \\
\frac{e^{\epsilon_1}}{e^{\epsilon_1}+1}\frac{1}{s-1+e^{\epsilon_2}} & \text{ if } u\neq NULL, v\neq NULL \text{ and } u\neq v \\
\frac{e^{\epsilon_1}}{e^{\epsilon_1}+1}\frac{e^{\epsilon_2}}{s-1+e^{\epsilon_2}} & \text{ if } u\neq NULL, v\neq NULL \text{ and } u= v
\end{array}
\right.
\end{equation}
Next, we have Theorem \ref{th:rr} to demonstrate the privacy guarantee of the RR mechanism.

\begin{theorem}
\label{th:rr}
The RR mechanism guarantees $(\epsilon_1+\epsilon_2)$-cell LDP.
\end{theorem}
\begin{proof}
Let $\mathcal{R}$ be the RR mechanism. We need to prove that for any two neighboring answer vectors $\vec{a_i}$ and $\vec{a_j}$, $\frac{P[\mathcal{R}(\vec{a_i}=\vec{z_p})]}{P[\mathcal{R}(\vec{a_j}=\vec{z_p})]}\leq exp(\epsilon_1 + \epsilon_2)$. Again, without loss of generality, we assume that $\vec{a_i}$ and $\vec{a_j}$ only differ at the first element. Thus, it remains for us to show that $\frac{P[\mathcal{R}(a_{i1})=z_{p1}]}{P[\mathcal{R}(a_{j1})=z_{p1}]}\leq exp(\epsilon_1 + \epsilon_2)$.
In other words, we need to prove that 
\[max_{v}\frac{max_u P_{uv}}{min_u P_{uv}} \leq exp(\epsilon_1 + \epsilon_2).\]

Considering the five cases in Equation (\ref{eq:uv2}), we discuss the following two cases based on the value of $v$.

\noindent 1) If $v=NULL$, we have $max_u P_{uv}=\frac{e^{\epsilon_1}}{1+e^{\epsilon_1}}$ (when $u=NULL$), and $min_u P_{uv}=\frac{1}{1+e^{\epsilon_1}}$ (when $u\neq NULL$). We can infer that 
\[\frac{max_u P_{uv}}{min_u P_{uv}}=e^{\epsilon_1} < exp^{\epsilon_1+\epsilon_2}.\]

\noindent 2) If $v\neq NULL$, the three possible values of $P_{uv}$ are $\frac{e^{\epsilon_1}}{e^{\epsilon_1}+1}\frac{e^{\epsilon_2}}{s-1+e^{\epsilon_2}}$, $\frac{e^{\epsilon_1}}{e^{\epsilon_1}+1}\frac{1}{s-1+e^{\epsilon_2}}$, and $\frac{1}{s}\frac{1}{e^{\epsilon_1}+1}$. Since $\epsilon_1, \epsilon_2>0$, it is easy to see that $\frac{e^{\epsilon_1}}{e^{\epsilon_1}+1}\frac{e^{\epsilon_2}}{s-1+e^{\epsilon_2}}>\frac{e^{\epsilon_1}}{e^{\epsilon_1}+1}\frac{1}{s-1+e^{\epsilon_2}}>\frac{1}{s}\frac{1}{e^{\epsilon_1}+1}$. 
Therefore, we have
\begin{equation*}
\begin{split}
\frac{max_u P_{uv}}{min_u P_{uv}} &= \left(\frac{e^{\epsilon_1}}{e^{\epsilon_1}+1}\frac{e^{\epsilon_2}}{s-1+e^{\epsilon_2}}\right)/\left(\frac{1}{s}\frac{1}{e^{\epsilon_1}+1}\right)\\
&=\frac{s}{s-1+e^{\epsilon_2}}\cdot e^{\epsilon_1+\epsilon_2}\\
&\le e^{\epsilon_1+\epsilon_2}
\end{split}
\end{equation*}
Based on the above reasoning, we prove that the RR mechanism satisfies $(\epsilon_1+\epsilon_2)$-cell LDP.
\end{proof}
}

\nop{
Another widely used way of acquiring $\epsilon$-differential privacy is to use random response \cite{wang2016using}. The basic idea of random response is to create a probability transfer matrix $P$. Suppose the original value of a particular cell in the answer vector is $u$, the probability $P[u][v]$ is a conditional probability that the cell transfer to a another value $v$, both $u$ and $v$ are in the answer domain. In particular, we set
\begin{equation*}
P[u][v]=\left\{
\begin{aligned}
\frac{e^{\epsilon}}{s-1+e^{\epsilon}}&\vspace{4cm}&v=u\\
\frac{1}{s-1+e^{\epsilon}}&\vspace{4cm}&v\not=u
\end{aligned}
\right.
\end{equation*}
The random mechanism $\mathcal{M}(\cdot)$ is to transfer all cells in the answer vector by fetching a value according to the probability matrix $P$. Then we can infer that given two neighboring answer vector $\vec{a}$ and $\vec{a'}$ that differ in one element, say the first element $a_1$, we have:
\begin{equation}
\label{equation:rr_original_proof}
\begin{aligned}
Pr\left[\frac{\mathcal{M}(\vec{a})=\vec{b}}{\mathcal{M}(\vec{a'})=\vec{b}}\right]\le\frac{\max_i P[i][b_1]}{\min_j P[j][b_1]}=\frac{\frac{e^{\epsilon}}{s-1+e^{\epsilon}}}{\frac{1}{s-1+e^{\epsilon}}}=e^{\epsilon}
\end{aligned}
\end{equation}

When we take NULL value into consideration, we need to be very careful. Since the mechanism is non-related to the answer domain and the size of domain may vary, we cannot simply treat NULL value as a part of the domain. Otherwise the larger domain size it has, the lower probability NULL value gets. Intuitively, NULL values must be considered independently as the answer domain. In other words, NULL value has equal weight as the whole answer domain. According to this idea, we designed the following mechanism.

We use a budget $\epsilon_1>0$ that adjusts the balance of random response between \emph{NULL}(denoted by $N$) and \emph{non-NULL} answer (denoted by $A$). The probability transfer matrix $P$ between them is:
\[
P=
\begin{pmatrix}
\frac{e^{\epsilon_1}}{e^{\epsilon_1}+1} & \frac{1}{e^{\epsilon_1}+1} \\
\frac{1}{e^{\epsilon_1}+1} & \frac{e^{\epsilon_1}}{e^{\epsilon_1}+1}
\end{pmatrix}
\]
the rows (denoted by $u$) are actual answer and columns (denoted by $v$) are random-responded answers. $u,v\in\{N,A\}$. When $u=v$, e.g., $P(v=N|u=A)$, the probability is $\frac{e^{\epsilon_1}}{e^{\epsilon_1}+1}$, otherwise it's $\frac{1}{e^{\epsilon_1}+1}$.

Next we construct the matrix $P'$ that expand the cases on $P(v=A|u=A)$. Here when another budget $\epsilon_2$ is picked to adjust the random response between non-NULL values. The probability transfer matrix $P'$ will be:
\[
P'=
\begin{pmatrix}
\frac{e^{\epsilon_2}}{s-1+e^{\epsilon_2}} & \frac{1}{s-1+e^{\epsilon_2}} & \hdots & \hdots & \frac{1}{s-1+e^{\epsilon_2}}\\
\frac{1}{s-1+e^{\epsilon_2}} & \frac{e^{\epsilon_2}}{s-1+e^{\epsilon_2}} & \frac{1}{s-1+e^{\epsilon_2}} & \hdots & \frac{1}{s-1+e^{\epsilon_2}}\\
\frac{1}{s-1+e^{\epsilon_2}} & \frac{1}{s-1+e^{\epsilon_2}} & \ddots &   & \frac{1}{s-1+e^{\epsilon_2}}\\
\vdots & \vdots &   & \ddots & \frac{1}{s-1+e^{\epsilon_2}}\\
\frac{1}{s-1+e^{\epsilon_2}} & \frac{1}{s-1+e^{\epsilon_2}} & \frac{1}{s-1+e^{\epsilon_2}} & \frac{1}{s-1+e^{\epsilon_2}} & \frac{e^{\epsilon_2}}{s-1+e^{\epsilon_2}}
\end{pmatrix}
\].
When $u=v$, the probability is $\frac{e^{\epsilon_2}}{s-1+e^{\epsilon_2}}$, otherwise it's $\frac{1}{s-1+e^{\epsilon_2}}$.

Since our goal is to have a random response mechanism that allows \emph{NULL} values. NULL values can become non-NULL values after the random response and vice versa. In that case, we define a $n+1$ by $n+1$ probability matrix by combining $P$ and $P'$ together as following and expand it:
\begin{equation}
\label{equation:nested_matrix}
P=
\begin{pmatrix}
\frac{e^{\epsilon_1}}{e^{\epsilon_1}+1} & \frac{1}{e^{\epsilon_1}+1} \\
\frac{1}{e^{\epsilon_1}+1} & \frac{e^{\epsilon_1}}{e^{\epsilon_1}+1}P'
\end{pmatrix}
\end{equation}
We will discuss all 5 cases below:
\begin{itemize}
\item when $u=NULL$ and $v=NULL$, \[P(v=NULL|u=NULL)=\frac{e^{\epsilon_1}}{e^{\epsilon_1}+1}\]
\item when $u=NULL$ and $v=x$, \[P(v=x|u=NULL)=\frac{1}{s}\frac{1}{e^{\epsilon_1}+1}\]
\item when $u=x$ and $v=NULL$, \[P(v=NULL|u=x)=\frac{1}{e^{\epsilon_1}+1}\]
\item when $u=x$, $v=y$ and $x\not=y$, \[P(v=y|u=x)=\frac{e^{\epsilon_1}}{e^{\epsilon_1}+1}\frac{1}{s-1+e^{\epsilon_2}}\]
\item when $u=x$, $v=y$ and $x=y$, \[P(v=y|u=x)=\frac{e^{\epsilon_1}}{e^{\epsilon_1}+1}\frac{e^{\epsilon_2}}{s-1+e^{\epsilon_2}}\]
\end{itemize}

\begin{theorem}
A random response mechanism following the probability matrix in equation (\ref{equation:nested_matrix}) above satisfies $(\epsilon_1+\epsilon_2)$-differential privacy.
\end{theorem}
\begin{proof}
According to the definition of differential privacy and the proof of original random response in equation (\ref{equation:rr_original_proof}), we just need to prove the following inequality:
\[\frac{P(v=b_1|u=a_1)}{P(v=b_1|u=a'_1)}=\frac{\max_i P(v=b_1|u=i)}{\min_j P(v=b_1|u=j)}\le e^{\epsilon}\]
where $a_1$ and $a'_1$ are the first element in two neighboring answer vector $\vec{a}$ and $\vec{a}'$. Here we prove it in two cases:\\

\noindent 1) When $b=NULL$, the probability that \[P(v=b|u=NULL)=\frac{e^{\epsilon_1}}{e^{\epsilon_1}+1},\] and for all other normal values $x$, we have \[P(v=b|u=x)=\frac{1}{e^{\epsilon_1}+1}.\] Apparently $\frac{e^{\epsilon_1}}{e^{\epsilon_1}+1}>\frac{1}{e^{\epsilon_1}+1}$ as $\epsilon_1>0$. Therefore,
\begin{equation*}
\begin{aligned}
\frac{\max_i P(v=b|u=i)}{\min_j P(v=b|u=j)}&=\frac{e^{\epsilon_1}}{e^{\epsilon_1}+1}/\frac{1}{e^{\epsilon_1}+1}\\
&=e^{\epsilon_1}\\
&\le e^{\epsilon_1+\epsilon_2}\\
\end{aligned}
\end{equation*}\\

\noindent 2) When $b$ is normal value, we have \[P(v=b|u=NULL)=\frac{1}{s}\frac{1}{e^{\epsilon_1}+1}.\] For all other normal values $x$, if $x\not=b$, we have \[P(v=b|u=x)=\frac{e^{\epsilon_1}}{e^{\epsilon_1}+1}\frac{1}{s-1+e^{\epsilon_2}}\] and if $x=b$, we have \[P(v=b|u=x)=\frac{e^{\epsilon_1}}{e^{\epsilon_1}+1}\frac{e^{\epsilon_2}}{s-1+e^{\epsilon_2}}\]. As $\epsilon_1,\epsilon_2>0$, we have $\frac{e^{\epsilon_1}}{e^{\epsilon_1}+1}\frac{e^{\epsilon_2}}{s-1+e^{\epsilon_2}}>\frac{e^{\epsilon_1}}{e^{\epsilon_1}+1}\frac{1}{s-1+e^{\epsilon_2}}>\frac{1}{s}\frac{1}{e^{\epsilon_1}+1}$. Therefore, we have:
\begin{equation*}
\begin{aligned}
\frac{\max_i P(v=b|u=i)}{\min_j P(v=b|u=j)}&=\left(\frac{e^{\epsilon_1}}{e^{\epsilon_1}+1}\frac{e^{\epsilon_2}}{s-1+e^{\epsilon_2}}\right)/\left(\frac{1}{s}\frac{1}{e^{\epsilon_1}+1}\right)\\
&=\frac{s}{s-1+e^{\epsilon_2}}\cdot e^{\epsilon_1+\epsilon_2}\\
&\le e^{\epsilon_1+\epsilon_2}
\end{aligned}
\end{equation*}
\end{proof}
}

  \vspace{-0.1in}
\section{Matrix Factorization (MF) Perturbation}
\label{sc:mf}
\vspace{-0.1in}
The existing works \cite{hua2015differentially,shin2018privacy} 
follow the same strategy: each worker computes the worker profile vector locally, without any interaction with DC. Then each worker   learns the task profile matrix from all the worker answers by iterative interactions with DC. 
Though correct, this may incur high communication cost. Therefore, we design a new protocol that does not require the communication between DC and the workers for matrix factorization under LDP. 

Initially, DC randomly generates a {\em task profile} matrix $V$, which is a $d\times n$ matrix whose values are generated independently from the worker answers, where $d$ is the factorization parameter. 
For each column $\vec{v}_j$ in $V$, we require its 1-norm is within 1, i.e., $\left\|\vec{v}_j\right\|_1 \leq 1$.
The purpose of this restriction is to provide $\epsilon$-cell LDP. 
 When $V$ is ready, DC sends it to each worker. This can be done when the workers accept the tasks. Then for any worker $W_i$ who has his answer vector $\vec{a}_i$ (i.e., a $1\times n$ answer matrix),  and the matrix $V$ from DC, he applies the MF method to compute the {\em worker profile} vector $\vec{u}_i$ that satisfies $\epsilon$-cell LDP by adding Laplace noise to the loss function (Equation (\ref{eq:loss})), i.e., 
 \begin{equation}
\label{eq:dp_loss}
L_{DP}(\vec{a}_i, \vec{u}_i, V) = \sum_{T_j\in \mathcal{T}_i} (a_{i,j}-\vec{u}_i^T \vec{v}_j)^2 + 2\vec{u}_i^T \vec{\eta}_i,
\end{equation}
where $\vec{\eta}_i=\{Lap(\frac{|\Gamma|}{\epsilon}), \dots, Lap(\frac{|\Gamma|}{\epsilon})\}$ is a $d$-dimensional vector. The perturbed worker profile vector $\vec{u_i}$ is computed as:
\begin{equation}
\label{eq:dp_u}
\vec{u}_i = \argmin_{\vec{u}_i} L_{DP}(\vec{a}_i, \vec{u}_i, V). 
\end{equation}
Based on the perturbed worker profile vector $\vec{u_i}$, the perturbed answer vector is computed as $\mathcal{M}(\vec{a}_i) = \vec{u}_i V$. Worker $W_i$ sends $\mathcal{M}(\vec{a}_i)$ to DC. 
Algorithm \ref{alg:loss_client} shows the pseudo code. We use gradient descent method (Line 3 - 5 of Algorithm \ref{alg:loss_client}) to compute $\vec{u}_i$. 

\nop{
The output $\vec{u}_i$ and $V$ must both satisfy $\epsilon$-DP. According to the sequential composition theorem \cite{mcsherry2009privacy}, the privatized answer vector $\vec{a}_i'=\vec{u}_iV$ also complies with $\epsilon$-DP.
However, this simple approach incurs significant computational complexity at the worker side. 
In particular, the time complexity of the above gradient descent procedure is $O(n(d+\ell))$, where $\ell$ is the number of iterations to reach convergence \Boxiang{Haipei, I believe that the complexity should be $O(nd\ell)$, since each iteration incurs $O(nd)$ complexity}. 
In order to provide high accuracy in the factorization result, $d$ cannot be significantly smaller than $n$. Therefore, the computational overhead can be overwhelming for the worker, especially when the number of tasks is large. 

To address the issue, our remedy is to let the client learn $\vec{u}_i$ under a fixed matrix factor $V$. 
Since it is the DC that collects the profile vectors $\{\vec{u}_i\}$ from the workers, we let the DC generates a random $d\times n$ matrix $V$ and sends it to the workers at the time of task assignment. 
After the worker produces the answer vector $\vec{a_i}$, he/she factorizes it by using the fixed profile matrix $V$ and obtains $\vec{u_i}$. This process takes $O(n\ell)$ time complexity. With an appropriate learning rate $\gamma$, $\ell$ can be a very small value. 
Thus, the computational cost at the worker side is affordable.
Next, we discuss the MF mechanism in details. 

\begin{algorithm}[!htbp]
\caption{Matrix factorization perturbation for DC}
\label{alg:loss_server}
\begin{algorithmic}[1]
    \REQUIRE The tasks $\mathcal{T}$, factorization parameter $d$
    \ENSURE Inferred truth $\{\hat{a_j}\}$
    
    \STATE Generate factor matrix $V$ s.t. each column vector $\left\|\vec{v}_j\right\|_1=1$
    \FOR{each worker $W_i\in \mathcal{W}$}
    	\STATE Send $\mathcal{T}_i$, $V$ and $d$ to $W_i$ \Wendy{who sends these information to $W_i$?}
    \ENDFOR
    \STATE Keep waiting until all workers report their factor vectors $\{\vec{u}_i\}$ and get factor matrix $U$
    \STATE $A'=UV$
    \STATE Infer truth $\{\hat{a}_j\}$ on $A'$ by following Algorithm \ref{alg:estimate}
    \RETURN $\{\hat{a}_j\}$
  \end{algorithmic}
\end{algorithm}
\Wendy{Algorithm 2 is straightforward. Don't need to include in the paper. }
Given a set of tasks $\mathcal{T}$ and the factorization parameter $d$, the DC produces a random $d\times n$ positive profile matrix $V$, where each column vector $\vec{v}_j$ satisfies $\left\|\vec{v}_j\right\|_1= 1$.
Next, the DC assigns tasks as well as $V$ to the workers.
After receiving the privatized factor vectors $\{\vec{u}_i\}$ from the workers, the DC recovers the perturbed answer matrix $A'=UV$ and applies truth inference over it. Algorithm \ref{alg:loss_server} shows the pseudo code. We use gradient descent method (Line 3 - 5 of Algorithm \ref{alg:loss_server}) to compute $\vec{u}_i$. 
}
\vspace{-0.1in}
\begin{algorithm}
\caption{Matrix factorization perturbation}
\label{alg:loss_client}
\begin{algorithmic}[1]
    \REQUIRE Factorization parameter $d$, privacy budget $\epsilon$, task profile matrix $V$, original answer vector $\vec{a}_i$. 
    \ENSURE Perturbed answer vector $\vec{a}_i$
    
    \STATE Randomly generate a $1\times d$ vector $\vec{u}_i$;
    \STATE Generate Laplace perturbation vector $\vec{\eta}_i$; 
	\REPEAT
    	\STATE $\vec{u}_i=\vec{u}_i-\gamma \nabla_{\vec{u}_i} L_{DP}$ ($\gamma$: the learning rate); 
	\UNTIL{$\nabla_{\vec{u}_i} L_{DP}=0$}
    \RETURN $\vec{a}_i = \vec{u}_i V$ as the perturbed worker answers;
  \end{algorithmic}
\end{algorithm}
\vspace{-0.1in}

Next, we present Theorem \ref{th:mf_dp} to formally prove the privacy guarantee. 
\begin{thm}
\label{th:mf_dp}
The MF mechanism guarantees $\epsilon$-cell LDP.

\noindent{\bf proof}
In order to prove that MF satisfies $\epsilon$-cell LDP on each answer, we first show that for any pair of answer vectors $\vec{a}_i$ and $\vec{a}_j$ that differ at one element,
\vspace{-0.1in}
\[\frac{Pr[\argmin_{\vec{u}_i} L_{DP}(\vec{a}_i,\vec{u}_i,V)=\vec{u}]}{Pr[\argmin_{\vec{u}_j} L_{DP}(\vec{a}_j,\vec{u}_j,V)=\vec{u}]}\le e^{\epsilon}.\] 
Without loss of generality, we assume that $\vec{a}_i$ and $\vec{a}_j$ differ at the first element. 
In Algorithm \ref{alg:loss_client}, the perturbed factor vector $\vec{u}_i$ is computed by requiring $\nabla_{\vec{u}_i} L_{DP}=0$. Therefore, we have: 

\resizebox{0.40\textwidth}{!}{
\begin{minipage}{\linewidth}
\begin{equation*}
\label{eq:loss_gradient}
\begin{aligned}
\nabla_{\vec{u}_i}L_{DP}(\vec{a}_i, \vec{u}_i, V) &= \sum_{T_k\in \mathcal{T}_i}\big[2(a_{ik}-\vec{u}_i^T\vec{v}_k)\cdot\nabla_{\vec{u}_i}(a_{ik}-\vec{u}_i^T\vec{v}_k) \big] + 2\vec{\eta}_i \\
&= \sum_{T_k\in \mathcal{T}_i}\big[2(a_{ik}-\vec{u}_i^T\vec{v}_k)(-\vec{v}_k)\big] + 2\vec{\eta}_i \\
&= 2\vec{\eta}_i - \sum_{T_k\in \mathcal{T}_i} 2\vec{v}_k(a_{ik}-\vec{u}_i^T \vec{v}_k).
\end{aligned}
\end{equation*}
\end{minipage}
}
Since it must be true that 
\vspace{-0.1in}
\begin{equation*}
\begin{aligned}
\nabla_{\vec{u}_i} L_{DP}(\vec{a}_i,\vec{u}_i,V)=\nabla_{\vec{u}_j} L_{DP}(\vec{a}_j,\vec{u}_j,V)=0.
\end{aligned}
\end{equation*}
\vspace{-0.1in}
We have:
\[2\vec{\eta}_i-\sum_{T_k\in \mathcal{T}_i} 2\vec{v}_k\left( a_{ik}-\vec{u}_i^T\vec{v}_k \right)=2\vec{\eta}_j-\sum_{T_k\in \mathcal{T}_j}2\vec{v}_k\left( a_{jk}-\vec{u}_j^T\vec{v}_k \right)\]
and
\[\sum_{T_k\in \mathcal{T}_i\cup \mathcal{T}_j} \vec{v}_k(a_{ik}-a_{jk})+(\vec{u}_j^T-\vec{u}_i^T)\vec{v}_k = \vec{\eta}_i-\vec{\eta}_j.\]
Since $\vec{u}_i=\vec{u}_j=\vec{u}$, we have:
\vspace{-0.1in}
\begin{equation*}
\begin{aligned}
\vec{\eta}_i-\vec{\eta}_j &= \sum_{T_k\in \mathcal{T}_i\cup \mathcal{T}_j} \vec{v}_k(a_{ik}-a_{jk}) \\
\left\|\vec{\eta}_i-\vec{\eta}_j\right\|_1&\le\left\|\vec{v}_1\right\|_1\left\|(a_{i,1}-a_{j,1})\right\|_1\\
\left\|\vec{\eta}_i-\vec{\eta}_j\right\|_1&\le|\Gamma|.
\end{aligned}
\end{equation*}
Now, we are ready to show that:
\begin{equation*}
\begin{aligned}
\frac{Pr[\argmin_{\vec{u}_i} L_{DP}(\vec{a}_i,\vec{u}_i,V)=\vec{u}]}{Pr[\argmin_{\vec{u}_j} L_{DP}(\vec{a}_j,\vec{u}_j,V)=\vec{u}]} &= \frac{Pr(\vec{\eta}_i)}{Pr(\vec{\eta}_j)} \\
&= \frac{\prod_{k=1}^d \exp\left(-\frac{\epsilon\cdot|\eta_{ik}|}{|\Gamma|}\right)}{\prod_{k=1}^d \exp\left(-\frac{\epsilon\cdot|\eta_{jk}|}{|\Gamma|}\right)} \\ 
&\leq \prod_{k=1}^{d} \exp\left(\frac{\epsilon\cdot|\eta_{jk}-\eta_{ik}|}{|\Gamma|}\right)\\
& = \exp\left(\frac{\epsilon\cdot\left\|\vec{\eta}_j-\vec{\eta}_i\right\|_1}{|\Gamma|}\right)\\
& \le \exp\left(\epsilon\right).
\end{aligned}
\end{equation*}
Therefore, $\vec{u}_i$ satisfies $\epsilon$-cell LDP. Since applying any deterministic function over a differentially private output still satisfies DP \cite{Mohammed:2011:DPD:2020408.2020487},  $\mathcal{M}(\vec{a}_i)=\vec{u}_iV$ also satisfies $\epsilon$-cell LDP. 
\end{thm}

Next, we have Theorem \ref{thm:mf_bound} to show the upper-bound of the expected error of the inferred truth of the MF approach. 
\begin{thm}
\label{thm:mf_bound}
Given a set of answer vectors $A = \{\vec{a}_i\}$, let $A^P = \{\hat{a}_i\}$ be the answer vectors after applying MF on $A$. The expected error $E\left[MAE(A^P)\right]$ of estimated truth based on the answer vectors perturbed by the MF mechanism satisfies that:
\[E\left[MAE(A^P)\right]\leq\tilde{q}m\left(\sqrt{\frac{2}{\pi}}+\frac{d|\Gamma|}{n\epsilon}\right),\] where $\tilde{q}=\max_i\{q_i\}$ and $d$ is the factorization parameter. 
\end{thm}
The proof of Theorem \ref{thm:mf_bound} can be found in our full paper \cite{sun2018truth}. 
Importantly, Theorem \ref{thm:mf_bound} shows that unlike LP and RR approaches, the error bound of the MF approach is insensitive to answer sparsity. This shows the advantage of the MF approach when dealing with sparse worker answers. 
Furthermore, the expected error of truth inference on the perturbed data is small. For example, consider the setting where $n=1,000$, $d=100$, $|\Gamma|=10$ and $\epsilon=1$. The expected error of the MF mechanism does not exceed 1.8, which is substantially smaller than that of the LP and RR approaches. Our empirical study also demonstrates the good utility of the MF mechanism in practice (more details in Section \ref{sc:exp}). 

\nop{
The basic idea of Laplace perturbation and random response are the same, they both perturb the value of cells directly. Due to this feature, NULL value must be specially take care of. Different from Laplace perturbation and random response, here we designed a new approach based on matrix factorization \Haipei{cite} called \emph{loss function perturbation} that can achieve the conditions of local differential privacy, and deal with NULL value naturally.

Matrix factorization approach is widely used in recommender system. It's a effective approach to predict the missing cells in the matrix and keep the consistency with the original value. In the follow: Suppose we have a $m$ by $n$ matrix $A$, we want to find two matrix $U$ and $V$, where $U$ is an $m$ by $k$ matrix and $V$ is a $k$ by $n$ matrix, that satisfy \[A\approx A'=UV\]. The solution is to use gradient descent to minimize the following loss function:
\begin{equation}
f(A,U,V)=\sum_{i=1}^m \sum_{j=1}^n (A_{i,j}-\sum_{r=1}^k U_{i,k}V_{k,j})^2
\end{equation}
and finally 
\begin{equation}
U,V=\argmin_{U,V}f(A,U,V)
\end{equation}

Following the idea of matrix factorization, our basic idea is to consider each single answer vector as a $1$ by $n$ matrix and factorize it as $\vec{a}=\vec{u}V$. We add Laplace perturbation on the loss function and $\vec{u}$ and $V$ will both satisfy $\epsilon$-differential privacy. Finally we multiply them together and get $\vec{a}'$. According to the sequential composition, $\vec{a}'$ will also satisfy differential privacy. 

However, in this case, matrix multiplication has a complexity as high as $O(n(k+l))$, where $l$ is the iteration time of gradient descent. Obviously, obviously $k$ cannot be too small, otherwise the accuracy of matrix factorization will drop dramatically. In that case, $O(n(k+l))$ will be too expensive for local client. Our solution is to fix one of the factor, the matrix $V$. Since the crowdsourcing platform will push tasks to the workers at the beginning and collect the answer vector in the end, we let the platform push a randomly generated matrix $V$ to each worker as well. After the workers finish their tasks, the factorize their answer vector locally and get $\vec{u}$. This step will on cost $O(nl)$. In most of the cases $l$ can be much smaller than $k$ as long as we appropriate learning rate. Finally the worker will only upload $\vec{u}$ to the platform, the the platform will do matrix multiplication instead. we will talk about the details in the following.

For the crowdsourcing platform, it will pick generates a matrix $k$ by $n$ matrix $V$, in which each column vector $\vec{v}_j$ satisfies $\left\|\vec{v}_j\right\|_1\le 1$.
After that, the DC assigns tasks to the workers, as well as $V$.. And push the matrix along with the tasks to each individual worker. Then it keeps waiting for workers' replies, i.e., the vectors $\vec{u}$. After getting the vectors $\vec{u}$ from all the workers, it will do matrix multiplication $U\times V$, where $U$ is the matrix that consist of all the vectors $\vec{u}$ and get the perturbed answer vectors of each individual worker that satisfy $\epsilon$-differential privacy. Finally the platform will infer truth from the perturbed answer vectors. The pseudo-code of the server part is shown in algorithm \ref{alg:loss_server}.

For the local clients, after getting the matrix $V$ along with the tasks, the worker will finish at lease one of the tasks and generate his answer vector $\vec{a}$. Then the worker will factorize and perturb his answer vector locally. Following the idea above, our loss function is:
\begin{equation}
f(\vec{a},\vec{u},V)=\sum_{j=1}^{n}\left( a_j-\vec{u} \vec{v}_j \right)^2
\end{equation}
where $V$ is a pre-generated constant from the crowdsourcing platform. Then we add Laplace perturbation on the loss function
\begin{equation}
f_{DP}(\vec{a},\vec{u},V)=\sum_{j=1}^{n}\left( a_j-\vec{u} \vec{v}_j \right)^2+2\vec{u}\vec{\xi}
\end{equation}
In this equation, $\vec{\xi}=\left\{Lap\left(\frac{\Delta}{\epsilon}\right)\right\}_k$ is a perturbation column vector with a dimension of $k$. And $\Delta$ is the range of answers, i.e., $\Delta=\max\{a_j\}-\min\{a_j\}$. The worker will use gradient descent to calculate perturbed $\vec{u}$ by following equation:
\begin{equation}
\label{equation:minimizing_loss_function}
\vec{u}=\argmin_{\vec{u}}f_{DP}(\vec{a},\vec{u},V)
\end{equation}
The pseudo-code is shown in algorithm \ref{alg:loss_client}. Next, we will give a formal proof that the mechanism satisfies $\epsilon$-differential privacy.

}

\begin{table*}[!htbp]
\centering
\begin{tabular}{|c|c|c|c|c|c|c|}
\hline
Dataset & \# of Workers & \# of Tasks & \# of Answers & Maximum Sparsity & Minimum Sparsity & Average Sparsity \\ 
\hline
Web&34&177&770&0.903955&0&0.705882 \\
\hline
AdultContent&825&11,040&89,796&0.999909&0.316033&0.993666\\
\hline
\end{tabular}
\caption{\label{table:realworld} Details of real-world datasets}
\vspace*{-0.25in}
\end{table*}

\nop{
\begin{figure}[!htbp]
\vspace{-0.05in}
\begin{center}
\begin{tabular}{@{}c@{}c@{}}
	\includegraphics[width=0.25\textwidth]{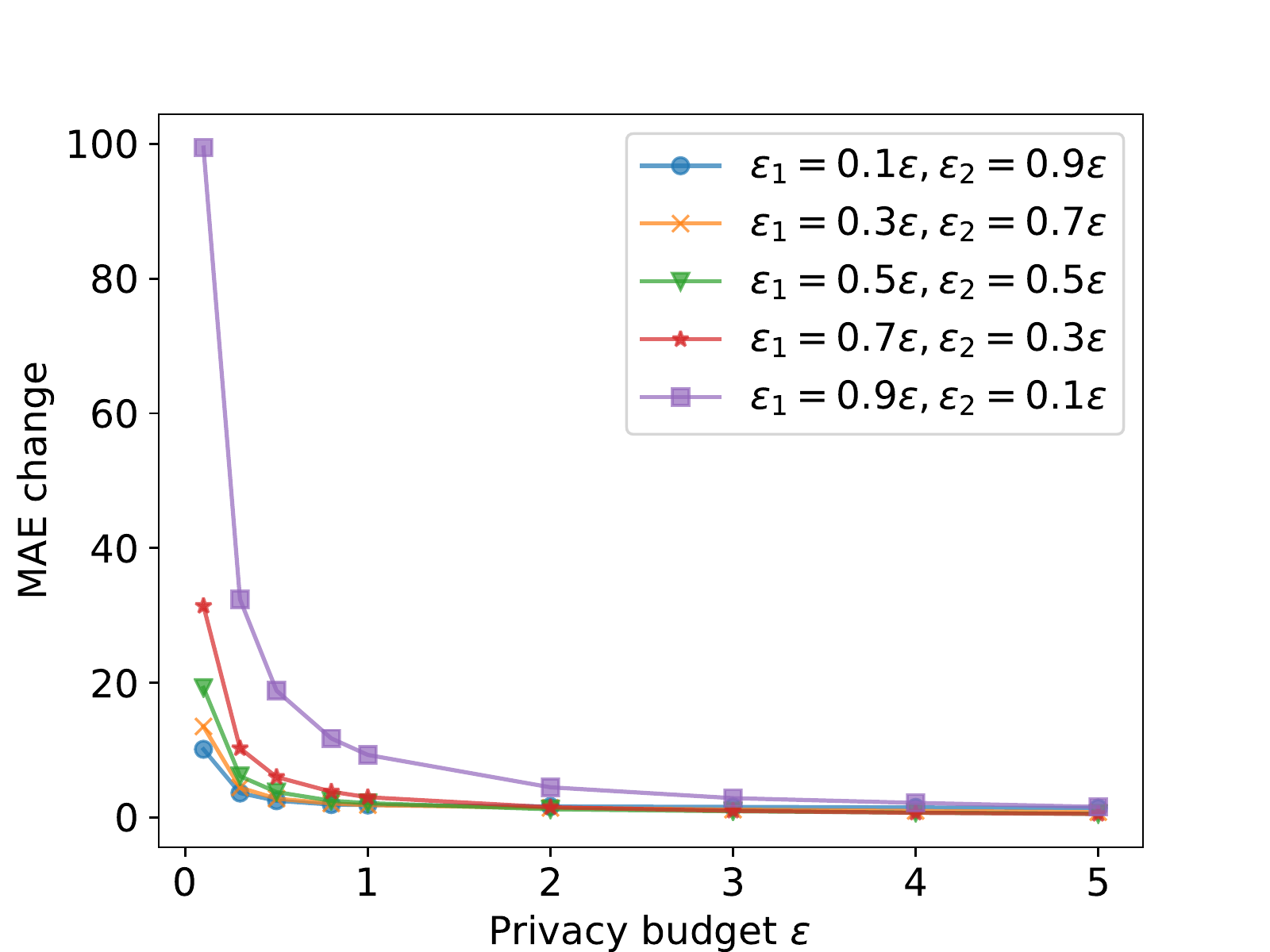} 
	&
	\includegraphics[width=0.25\textwidth]{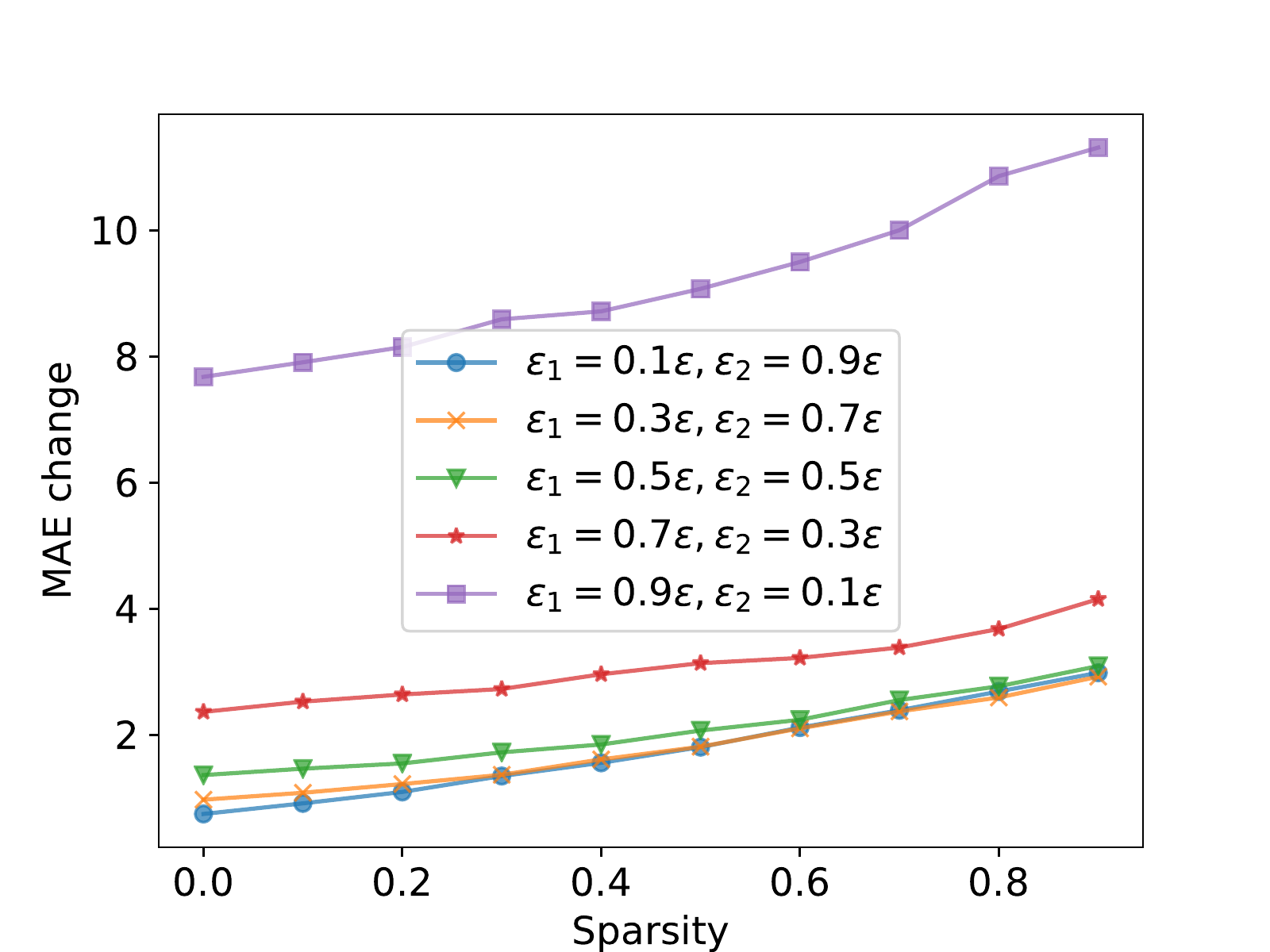}
    \\
	{\scriptsize(a) Various $\epsilon$}
	&
	{\scriptsize(b) Various sparsity}
\end{tabular}
\vspace{-0.1in}
 \caption{\small \label{fig:rrlp_parameter_assignment} Accuracy of truth inference with different privacy budget assignment strategy of RR+LP ( synthetic dataset, 200 workers and 20 tasks, truth distribution $\mathcal{N}(0,1)$)}
\vspace{-0.2in}
\end{center}
\end{figure}
}

\vspace{-0.1in}
\section{Experiments}
\label{sc:exp}
\vspace{-0.1in}
\vspace{-0.1in}
\subsection{Setup}
\label{section:exp_setup}
\vspace{-0.1in}

\noindent{\bf Synthetic datasets.} We generate several synthetic datasets for evaluation. The answer domain $\Gamma$ of these  datasets includes the integers from 0 to 9. For each task, we generate its ground truth by following $\mathcal{N}(0,1)$, i.e., the ground truth centers at answer 0. 
By following the assumption of the truth inference algorithm, for each worker $W_i$, his answers are generated by adding the Gaussian noise $\mathcal{N}(0,\sigma_i^2)$ on the ground truth, where $\sigma_i$ is decided by the worker quality. We define two types of workers, namely {\em high-quality} workers with $\sigma_i=1$ and {\em low-quality} workers with $\sigma_i=5$. We pick 50\% of the workers randomly as high-quality and the rest as low-quality. 

\noindent{\bf Real-world datasets.} We use two real-world crowdsourcing datasets: \emph{Web} dataset and \emph{AdultContent} dataset from a public data repository\footnote{\url{http://dbgroup.cs.tsinghua.edu.cn/ligl/crowddata/}}. In the \emph{Web} dataset, 34 workers provide the relevance score (from 0 to 4) of 177 pairs of URLs. In the \emph{AdultContent} dataset, 825 workers assign a score (from 0 to 4) of adult content of 11 thousand websites. 
Both real-world datasets have the ground truth of tasks available. More details of the real-world datasets are included in Table \ref{table:realworld}.

  
\noindent{\bf Utility metric and parameters.}
We measure the {\em MAE change} as the utility metrics. That is, we compare the MAE of the truth inference results derived from the original answers and that from the perturbed answers. Formally, let $MAE_{O}$ and $MAE_{P}$ be the MAE of the truth inference results before and after applying LDP on the worker answers. Then the MAE change $MAE_{C}$ is measured as $MAE_{C} = MAE_{P} - MAE_{O}$. The smaller the MAE change, the better the utility. 
We study the impacts of different parameters on $MAE_C$, including the size of the dataset, the privacy budget $\epsilon$, and the density of worker answers.
\nop{
To measure the impact of different parameters on $MAE_C$, we change various parameter settings, including the size of dataset, privacy budget $\epsilon$, and density of worker answers. }


\noindent{\bf Compared method.} We compare the performance of our perturbation mechanism with the \emph{2-Layer} approach \cite{li2018an}, which is the most relevant work to ours. The 2-layer approach assumes the data is complete. It relies on sampling and randomized response to realize LDP. We extend the 2-layer approach to make it deal with NULL values by treating the NULL value as a unique answer. Each worker samples his own probability for NULL values. 





\vspace{-0.1in}
\subsection{Distribution Analysis of Real-world Datasets}
\label{sc:realworld_analysis}
\vspace{-0.1in}
We analyze the answer sparsity and worker answer distribution of the two real-world datasets that we used for the experiments. 

{\bf Data sparsity.} We observe that both real-world datasets are very sparse. The sparsity is measured as the fraction of NULL values of each worker answer vector. In Table \ref{table:realworld}, we report the minimum, maximum, and average sparsity of the two real-world datasets. In particular, the \emph{AdultContent} dataset is extremely sparse. The average sparsity is greater than 0.99, while the maximum sparsity is as high as 0.999909 (i.e., the worker only answers one task). 

\begin{figure}[!htbp]
\begin{center}
\begin{tabular}{@{}c@{}c@{}}
    \includegraphics[width=0.25\textwidth]{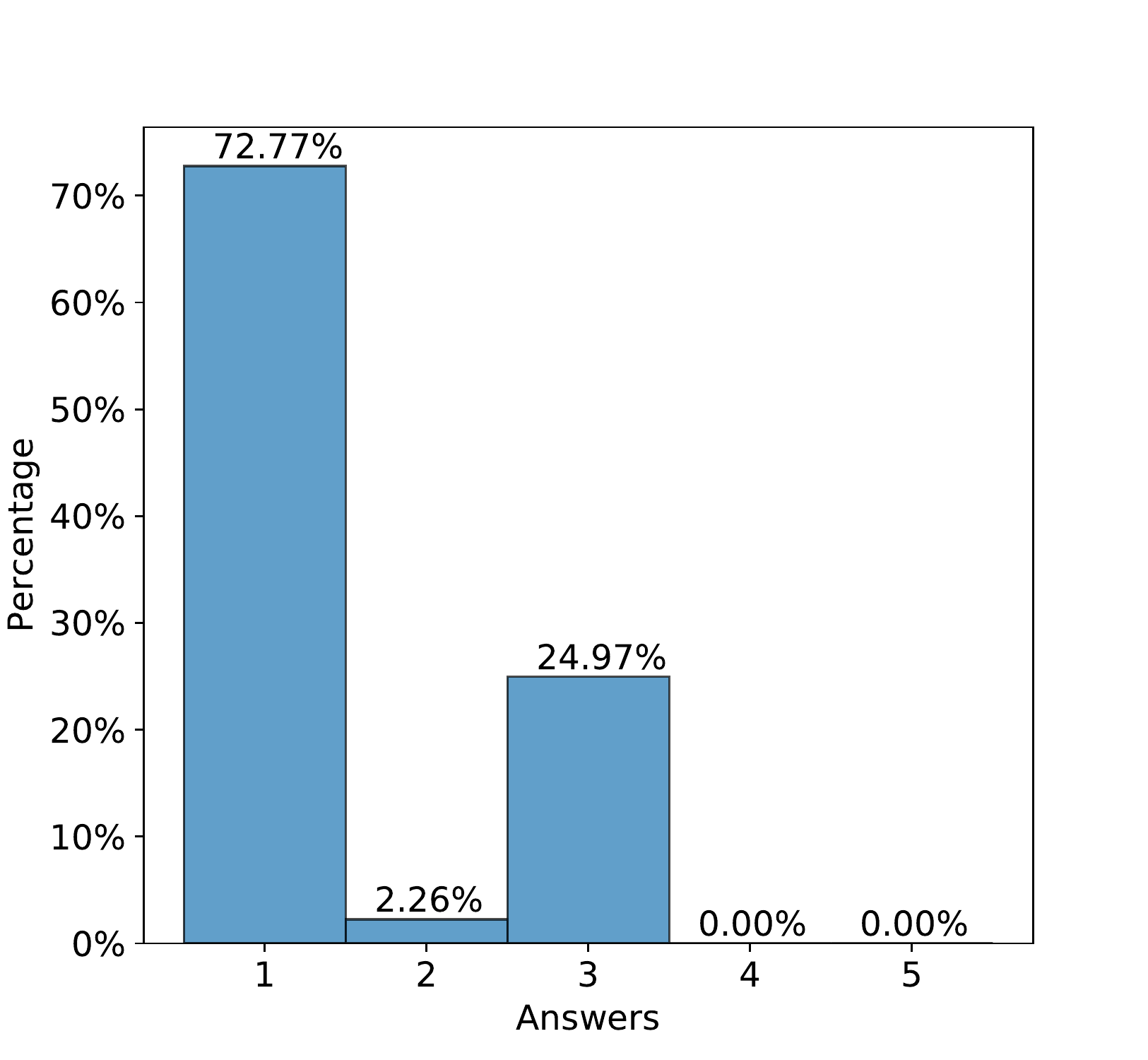} 
	&
	\includegraphics[width=0.25\textwidth]{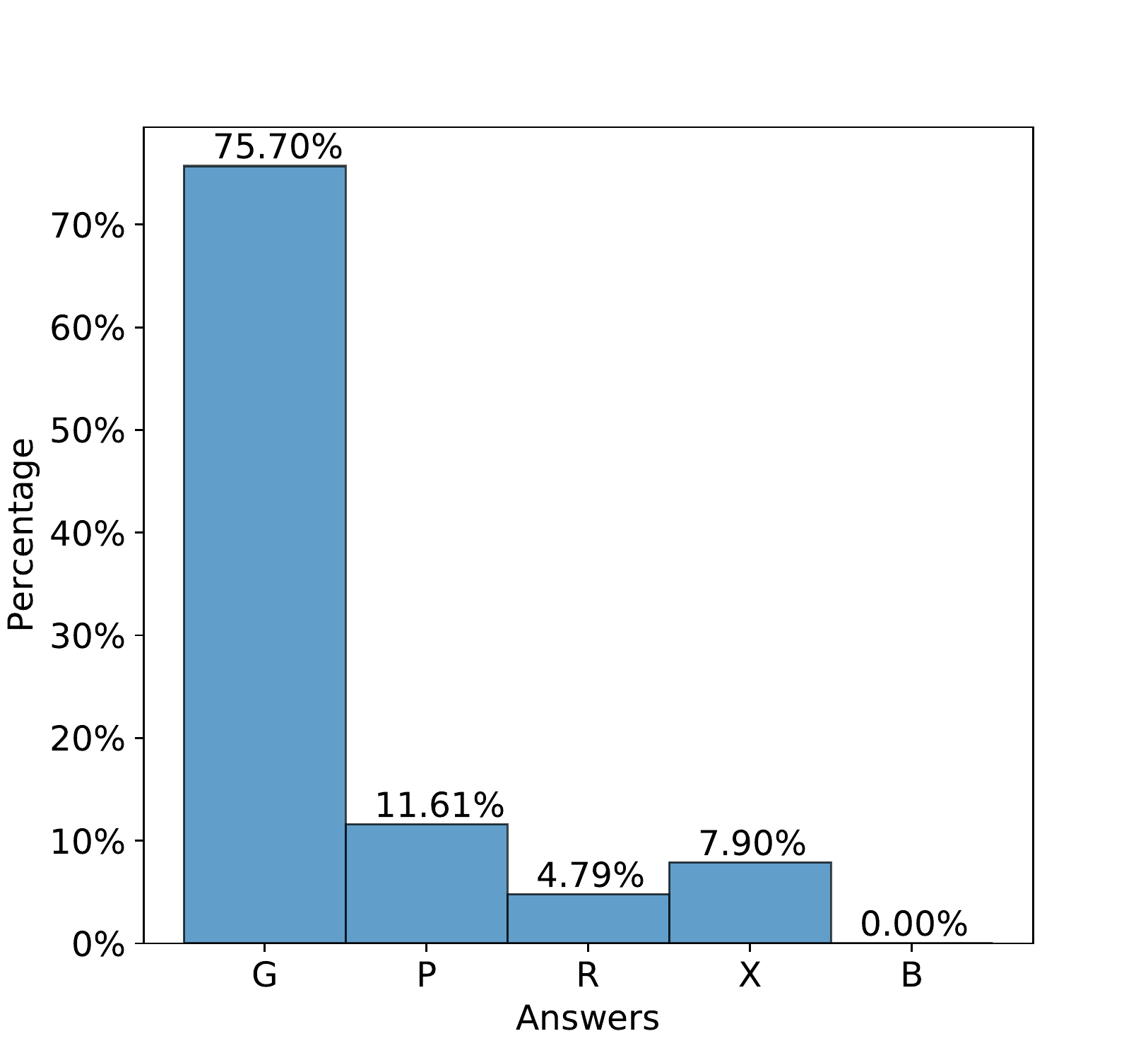}
	\\
	{\scriptsize(a) Web dataset}
	&
	{\scriptsize(b) AdultContent dataset}
\end{tabular}
\vspace{-0.1in}
  \caption{\small \label{fig:original_answer_distribution} Answer distribution of the real-world datasets}
\end{center}
\vspace*{-0.2in}
\end{figure}

\noindent{\bf Distributions of worker answers.} 
Figure \ref{fig:original_answer_distribution} shows the answer distribution of the two real-world datasets. The analysis shows that the worker answer distribution of both datasets is skewed. 
In both {\em Web} and the {\em AdultContent} datasets, the ground truth of more than 72\% tasks is value 0. Our analysis (Figure \ref{fig:original_answer_distribution} (a) and (b)) shows that the majority of answers in theses two datasets are correct (i.e., the same as the ground truth). Furthermore, the remaining worker answers are distributed to a small number of values in a non-uniform fashion.  



\nop{
\begin{figure*}[!htbp]
\begin{center}
\begin{tabular}{@{}c@{}c@{}c@{}}
	\includegraphics[width=0.33\textwidth]{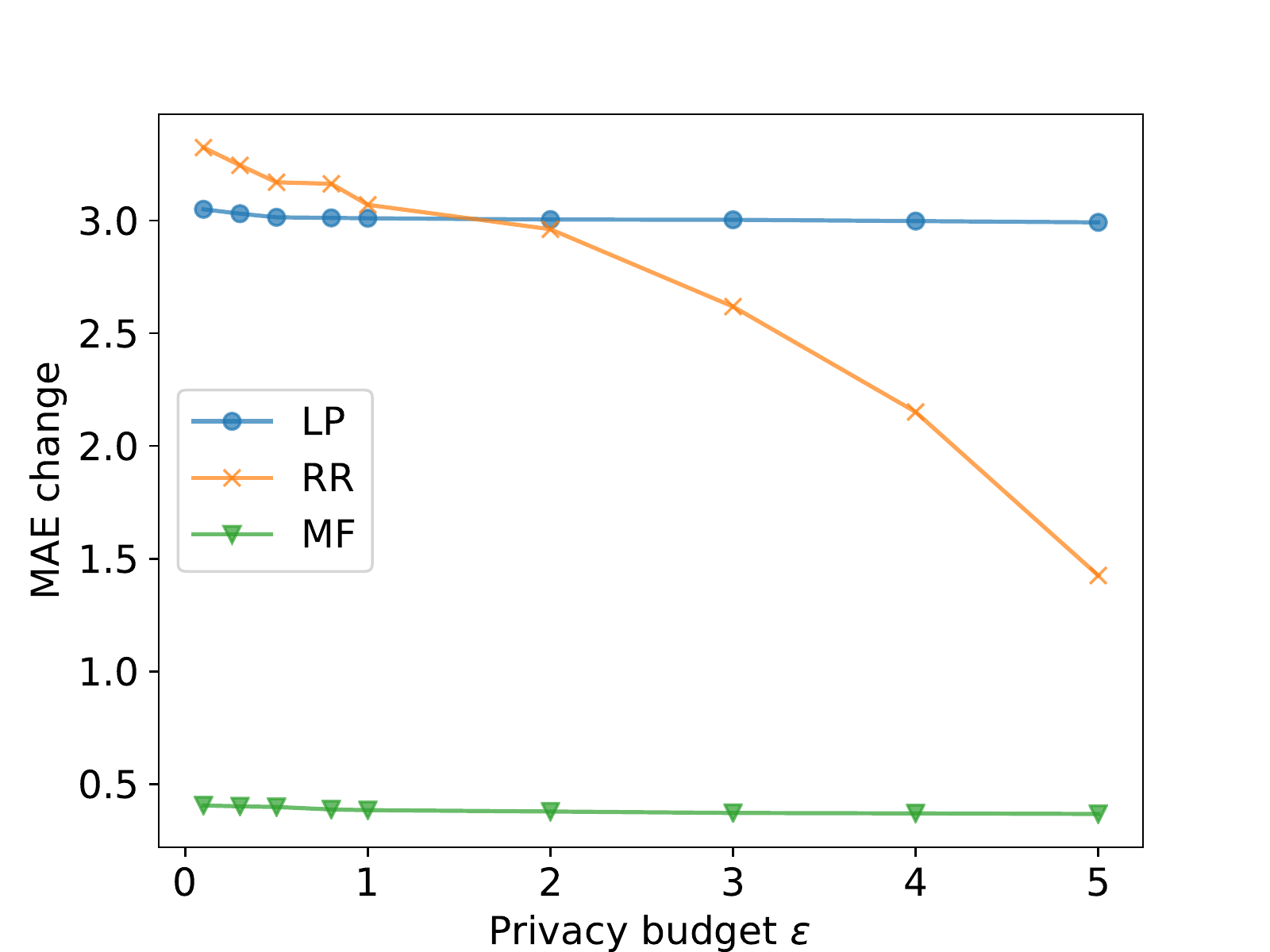} 
	&
	\includegraphics[width=0.33\textwidth]{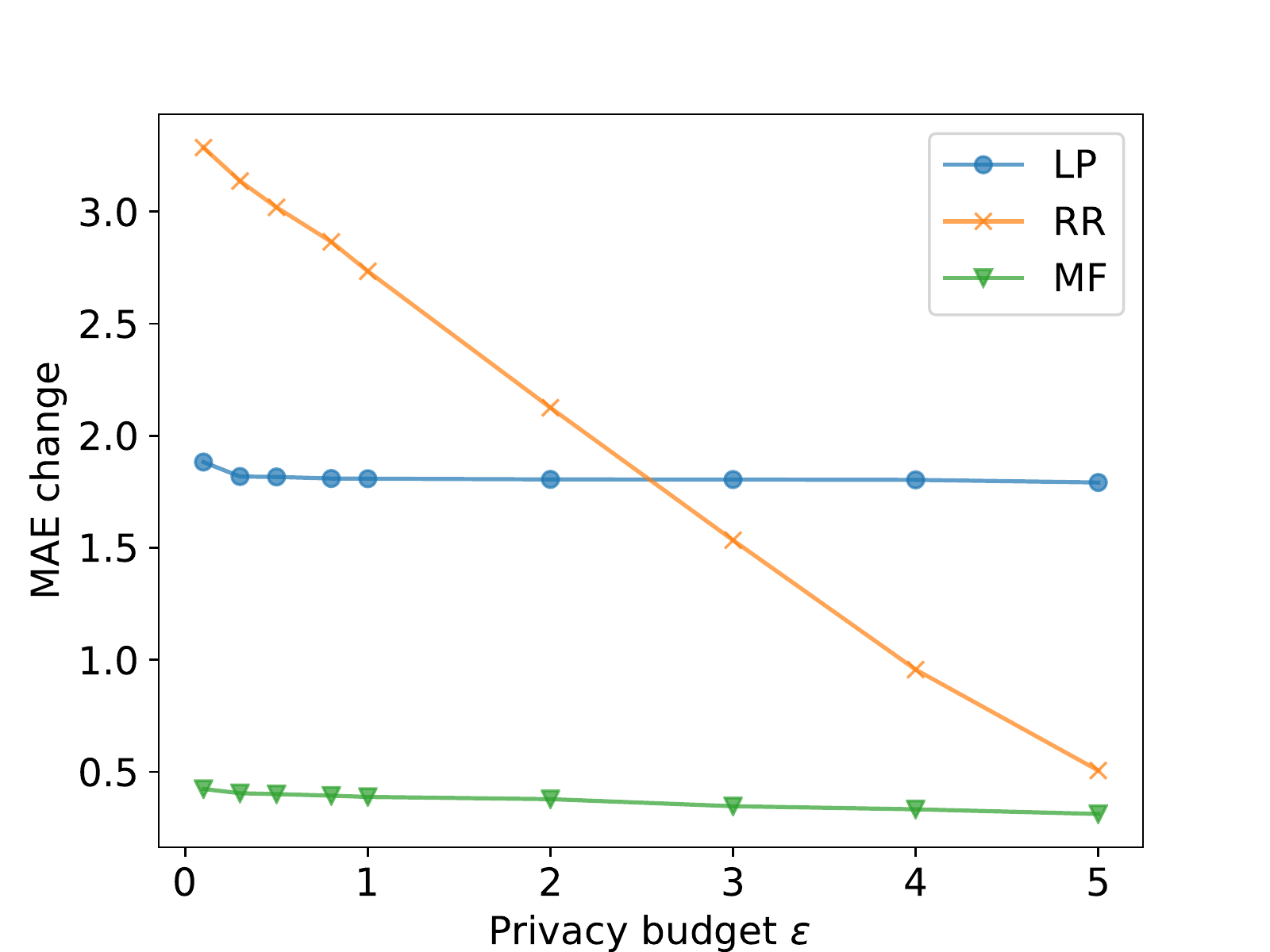}
	&
	\includegraphics[width=0.33\textwidth]{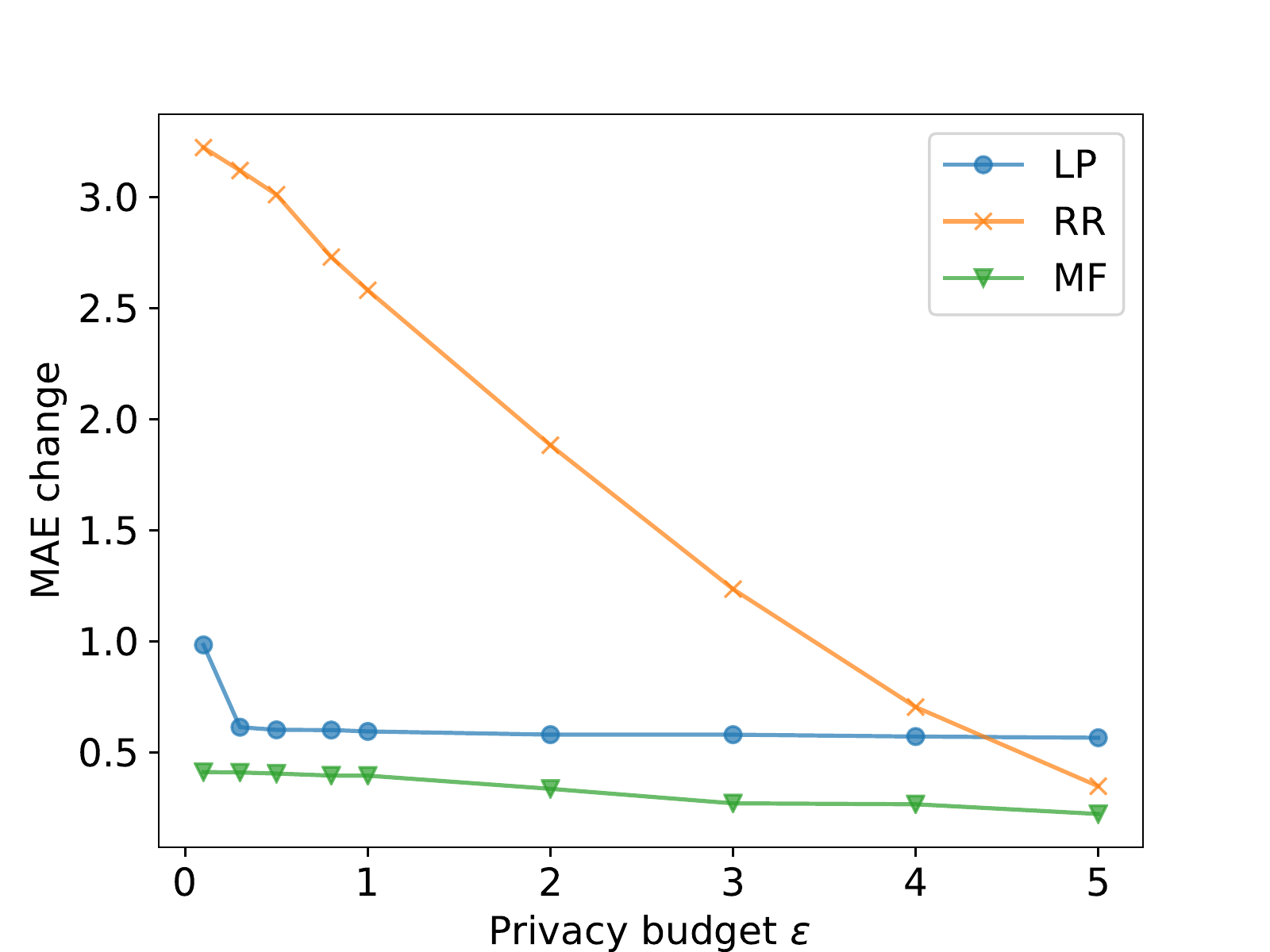}
	\\
	{\scriptsize(a) sparsity$=0.9$, truth distribution $\mathcal{N}(0,1)$}
	&
	{\scriptsize(b) sparsity$=0.5$, truth distribution $\mathcal{N}(0,1)$}
	&
	{\scriptsize(c) sparsity$=0.1$, truth distribution $\mathcal{N}(0,1)$}
	\\
    \includegraphics[width=0.33\textwidth]{./Figures/m10000n1000den0_10skew5_comp4_varyEpsilon.pdf} 
	&
	\includegraphics[width=0.33\textwidth]{./Figures/m10000n1000den0_50skew5_comp4_varyEpsilon.pdf}
	&
	\includegraphics[width=0.33\textwidth]{./Figures/m10000n1000den0_90skew5_comp4_varyEpsilon.pdf}
	\\
	{\scriptsize(d) sparsity$=0.9$, truth distribution $\mathcal{N}(4,1)$}
	&
	{\scriptsize(e) sparsity$=0.5$, truth distribution $\mathcal{N}(4,1)$}
	&
	{\scriptsize(f) sparsity$=0.1$, truth distribution $\mathcal{N}(4,1)$}
\end{tabular}
\vspace{-0.1in}
\caption{\small \label{fig:ae_wrt_epsilon_10000} Accuracy of truth inference w.r.t. different privacy budget $\epsilon$ (10000 workers and 1000 tasks).}
\end{center}
\end{figure*}
}

\nop{
\subsection{Privacy Budget Assignment for RR+LP}

Recall that the RR+LP mechanism follows a two-step process to provide $(\epsilon_1+\epsilon_2)$-LDP (Section \ref{sc:rrlp}), where $\epsilon_1$ and $\epsilon_2$  are the budget allocated to RR and LP respectively. 
We examine the MAE change of truth inference algorithm on RR+LP with various privacy budget assignment schemes for $\epsilon_1$ and $\epsilon_2$.
Part of the results are displayed in Figure \ref{fig:rrlp_parameter_assignment}.
 Intuitively, a larger $\epsilon_2$ value leads to smaller noise in the perturbed answer vectors and better utility. This intuition is consistent with our results. In particular, RR+LP delivers the smallest MAE change when $\epsilon_1=0.1\epsilon$ and $\epsilon_2=0.9\epsilon$, regardless of data sparsity and $\epsilon$. The remaining results are similar to our findings in Figure \ref{fig:rrlp_parameter_assignment} and thus are omitted due to limited space. Therefore, in the remaining experiments, we use $\epsilon_1=0.1\epsilon$ and $\epsilon_2=0.9\epsilon$ as the privacy budget configuration for the RR+LP mechanism.
}

\vspace{-0.1in}
\subsection{Parameter Impact on Truth Inference Accuracy} 
\label{sc:impact}
\vspace{-0.1in}
In this section, we present the impact of sparsity, data distribution, and data size on the accuracy of truth inference results. 

\vspace{-0.1in}
\begin{figure}[!htbp]
\begin{center}
\begin{tabular}{@{}c@{}c@{}}
	\includegraphics[width=0.25\textwidth]{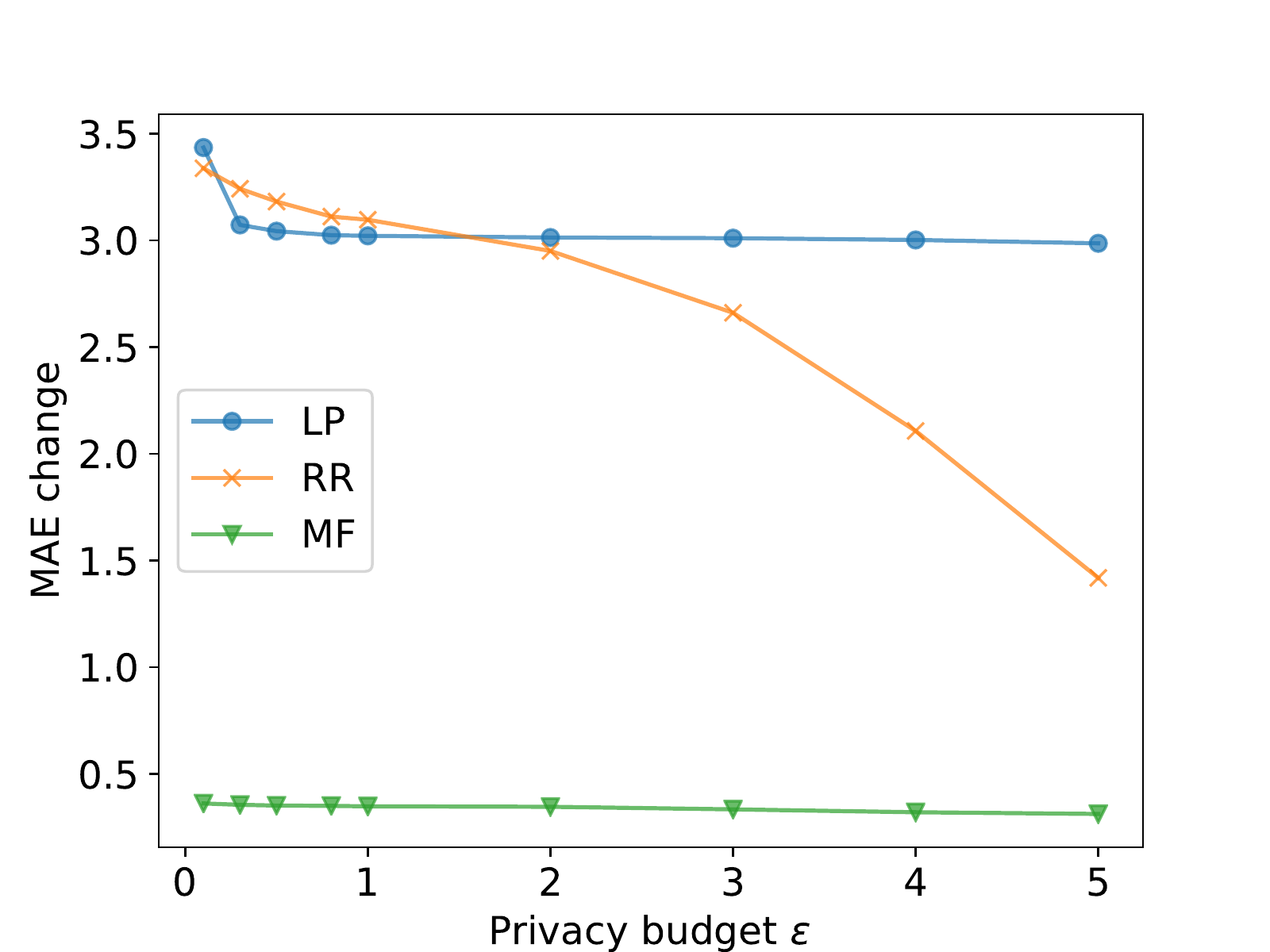} 
	&
	\includegraphics[width=0.25\textwidth]{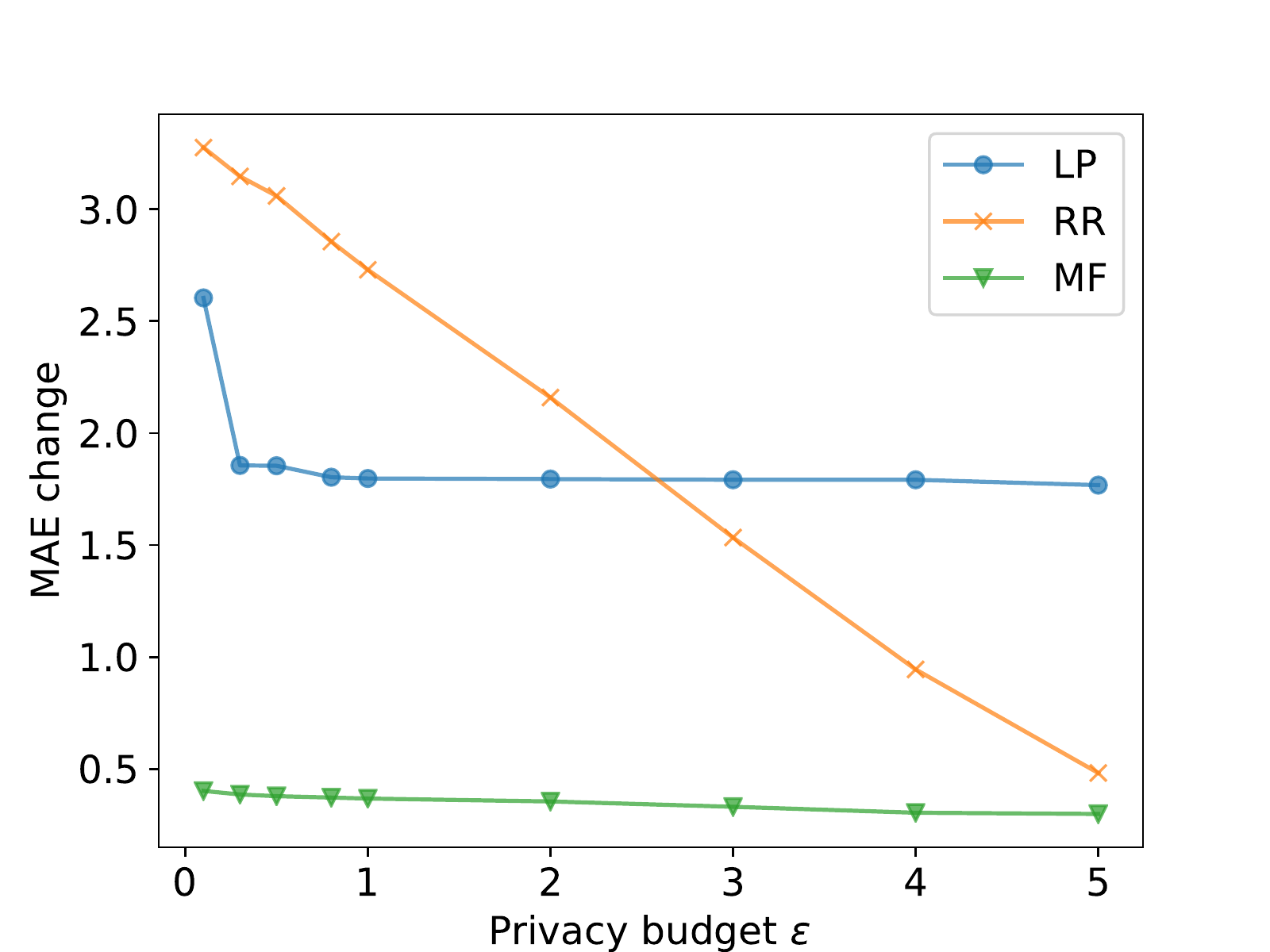}
	\\
	{\scriptsize(a) sparsity$=0.9$}
	&
	{\scriptsize(b) sparsity$=0.5$}
\nop{    
	\\
    \includegraphics[width=0.33\textwidth]{./Figures/m2000n200den0_10skew5_comp4_varyEpsilon.pdf} 
	&
	\includegraphics[width=0.33\textwidth]{./Figures/m2000n200den0_50skew5_comp4_varyEpsilon.pdf}
	&
	\includegraphics[width=0.33\textwidth]{./Figures/m2000n200den0_90skew5_comp4_varyEpsilon.pdf}
	\\
	{\scriptsize(d) sparsity$=0.9$, truth distribution $\mathcal{N}(4,1)$}
	&
	{\scriptsize(e) sparsity$=0.5$, truth distribution $\mathcal{N}(4,1)$}
	&
	{\scriptsize(f) sparsity$=0.1$, truth distribution $\mathcal{N}(4,1)$}
}    
\end{tabular}
\vspace{-0.1in}
\caption{\small \label{fig:ae_wrt_epsilon_200} Accuracy of truth inference w.r.t. different privacy budget $\epsilon$ (2000 workers and 200 tasks).}
\end{center}
\end{figure}

\noindent{\bf Impact of privacy budget.} First, we vary the privacy budget $\epsilon$, and measure MAE change of RR, LP, and MF. Figure \ref{fig:ae_wrt_epsilon_200} (a) - (b) show the results on the tasks whose truth follows the distribution $\mathcal{N}(0,1)$. 
Overall, MF always provides the smallest MAE change. 
The MAE change of LP witnesses a sudden drop when $\epsilon$ increases from 0.1 to 1, and keeps stable hereafter. This is because the Laplace noise is substantial when $\epsilon$ is small. 
On the other hand, the MAE change of RR keeps decreasing sharply with the growth of $\epsilon$.  This is because according to Equation (\ref{eq:rr}), when $\epsilon$ is as small as 0.1, $e^{\epsilon}\approx 1$. So the probability that RR keeps an answer unchanged is close to the probability that it substitutes the answer with a different random value in the domain. In other words, the perturbed answers are generated by following a uniform distribution. When $\epsilon$ increases, the existing values are much more likely to be kept in the dataset than being replaced. This leads to sharp drop of MAE change. 
Furthermore, when comparing Figure \ref{fig:ae_wrt_epsilon_200} (a) and (b), the MAE change of LP decreases when the sparsity goes down. This is because LP replace NULL values by following a uniform distribution over the domain [0,9] with the mean value 4.5, which is far from the center of answers 0. Thus, it leads to larger MAE change on sparse data. In the contrast, MF is not affected by the change of sparsity. This is consistent with our theoretical analysis.
\nop{
In Figure \ref{fig:ae_wrt_epsilon_200} (d) - (f), we show the MAE change of the inferred truth, where the real truth follows the distribution $\mathcal{N}(4,1)$ (i.e., the truth centered at value 4). With the shift of truth and consequently the distribution of worker answers, we observe the following interesting patterns of MAE change. 
First, the benefit of MF over the other approaches in Figure 3 (d) is not as significant as shown in Figure 3 (a).
The reason is that LP, RR and RR+LP replace the NULL answers with the answers that are centered at the value 4.5.  
This leads to relatively small MAE change for the truth distribution $\mathcal{N}(4,1)$ (i.e., centered at value 4). 
In addition, the latter reason leads to other two phenomenons. 
First, compared with the previous setting that the truth is centered at value 0, the MAE change of RR in this setting (truth centered at value 4) dramatically decreases. 
It even can win MF when $\epsilon$ is higher than 4, because when $\epsilon$ is larger than 4, with at least 85\% probability RR keeps the original answer unaltered, which further reduces the MAE change. 
Second, the MAE change of LP, RR, and RR+LP in Figure \ref{fig:ae_wrt_epsilon_200} (d) - (f) is not heavily impacted by sparsity. 
This is surprising, as these three approaches are expected to be influenced by sparsity (see our theoretical analysis). 
We analyze the reason behind this phenomenon. It turned out that the insensitivity to sparsity of LP, RR, and RR+LP only appears for this setting, due to the fact that the NULL values are replaced with the values whose mean is 4.5, which is very close to where the truth is centered.

\nop{We further compare Figure \ref{fig:ae_wrt_epsilon_200} (e) with Figure \ref{fig:ae_wrt_epsilon_200} (b), and notice that the MAE change of LP and RR+LP decrease at a faster pace in Figure \ref{fig:ae_wrt_epsilon_200} (e)  than Figure \ref{fig:ae_wrt_epsilon_200} (b) when $\epsilon$ increases from 0.1 to 1. This is because the MAE change of LP and RR+LP in Figure \ref{fig:ae_wrt_epsilon_200} (b) arise from both replacement of NULL values and Laplace noise. While in Figure \ref{fig:ae_wrt_epsilon_200} (e), the MAE change is mainly incurred by the Laplace noise. Therefore, for this case, LP and RR+LP is more sensitive to the variation of $\epsilon$.RR has the opposite pattern in Figure \ref{fig:ae_wrt_epsilon_200} (b) and (d). When the answers are centered at value 4, the original answers are more similar to uniform distribution. Hence, even a larger $\epsilon$ means higher probability to keep an answer unaltered rather than replacing it by a random value, the reduction in MAE change is limited. 
We also observe that in Figure \ref{fig:ae_wrt_epsilon_200} (e), LP starts with the same MAE change as Figure \ref{fig:ae_wrt_epsilon_200} (b) (around 2.5 for both figures), but sharp decreases (e.g., drops from 1.8 to 0.9 when $\epsilon=0.5$. }
For same reason, we witness smaller variation in the MAE change of RR when comparing Figure \ref{fig:ae_wrt_epsilon_200} (e) with (b).
In contrast, when $\epsilon=0.1$, the MAE change of RR and RR+LP remains large. When $\epsilon$ is small, the substantial Laplace noise makes the truth inference results considerably erroneous. 
}

\nop{
Another important observation is that, as shown in Figure \ref{fig:ae_wrt_epsilon_200} (f), the MAE change of MF becomes worse than that of RR, LP, and RR+LP when $\epsilon$ is 2 and higher. These are the only cases that MF loses to the straw-man approaches. 
This is because when $\epsilon$ is 2 and higher, the straw-man approaches either keep almost all the original answers unchanged or replace the NULL values with the values that are close to the truth.
\nop{
\Wendy{4. Compare figure 3 (d), (e), and (f), why MAE change of all approaches in 3 figures do not change much when sparsity decreases? This is different from Figure 3 (a) - (c), in which the MAE change of each specific approaches indeed changes when the sparsity changes.}
\Boxiang{Actually only RR has this result. The reason is the same as the discussion in Figure \ref{fig:ae_wrt_epsilon_200} (e). That is because the original answers are close to uniform distribution. So the replaced answers and original answers are very similar.}

LP and RR+LP result in the largest MAE when $\epsilon$ is small. This is because of the large Laplace noise introduced by these two mechanisms. However, when $\epsilon$ is larger than 1, the utility of LP and RR+LP is much better. The reason is similar to the discussion in Figure \ref{fig:ae_wrt_epsilon_200} (a) - (c).
In general, RR and MF produce the smallest MAE, regardless of the privacy budget and sparsity.  Recall that RR uniformly randomly maps an original rating to the answer domain. The average result of the perturbed answers is close to 4, which is the ground truth of most tasks. 
We also note that when $\epsilon$ is larger than 4, RR can slightly outperform MF. This is because when $\epsilon$ is large, RR keeps almost all answers unaltered in the perturbed data.
}
We also generate tasks whose truth centered at value 9. The observation is similar to that in Figure \ref{fig:ae_wrt_epsilon_200} (a) - (c). We omit the result due to the space limit. 
}

\vspace{-0.1in}
\begin{figure}[!htbp]
\begin{center}
\begin{tabular}{@{}c@{}c@{}}
	\includegraphics[width=0.25\textwidth]{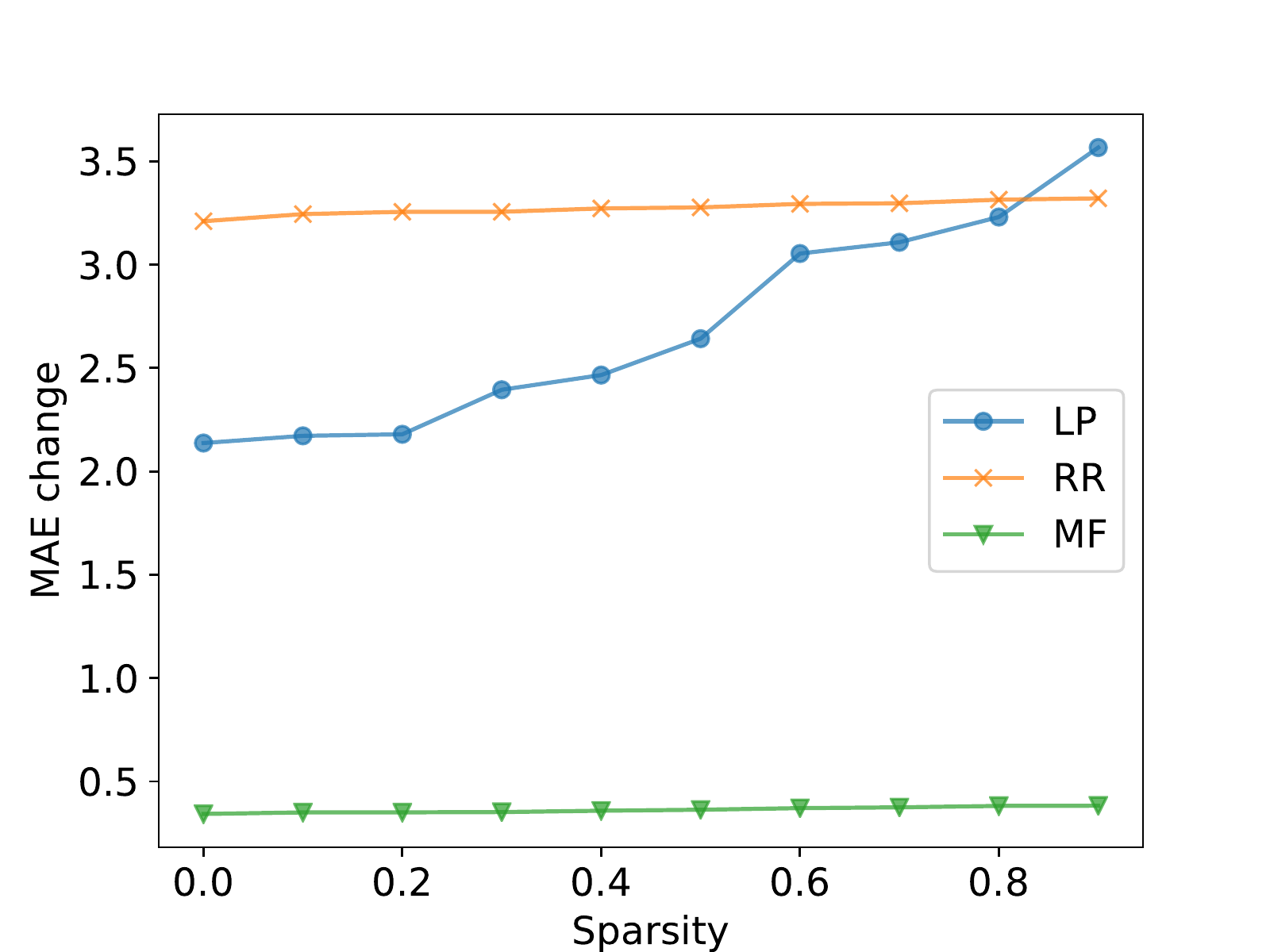} 
	&
	\includegraphics[width=0.25\textwidth]{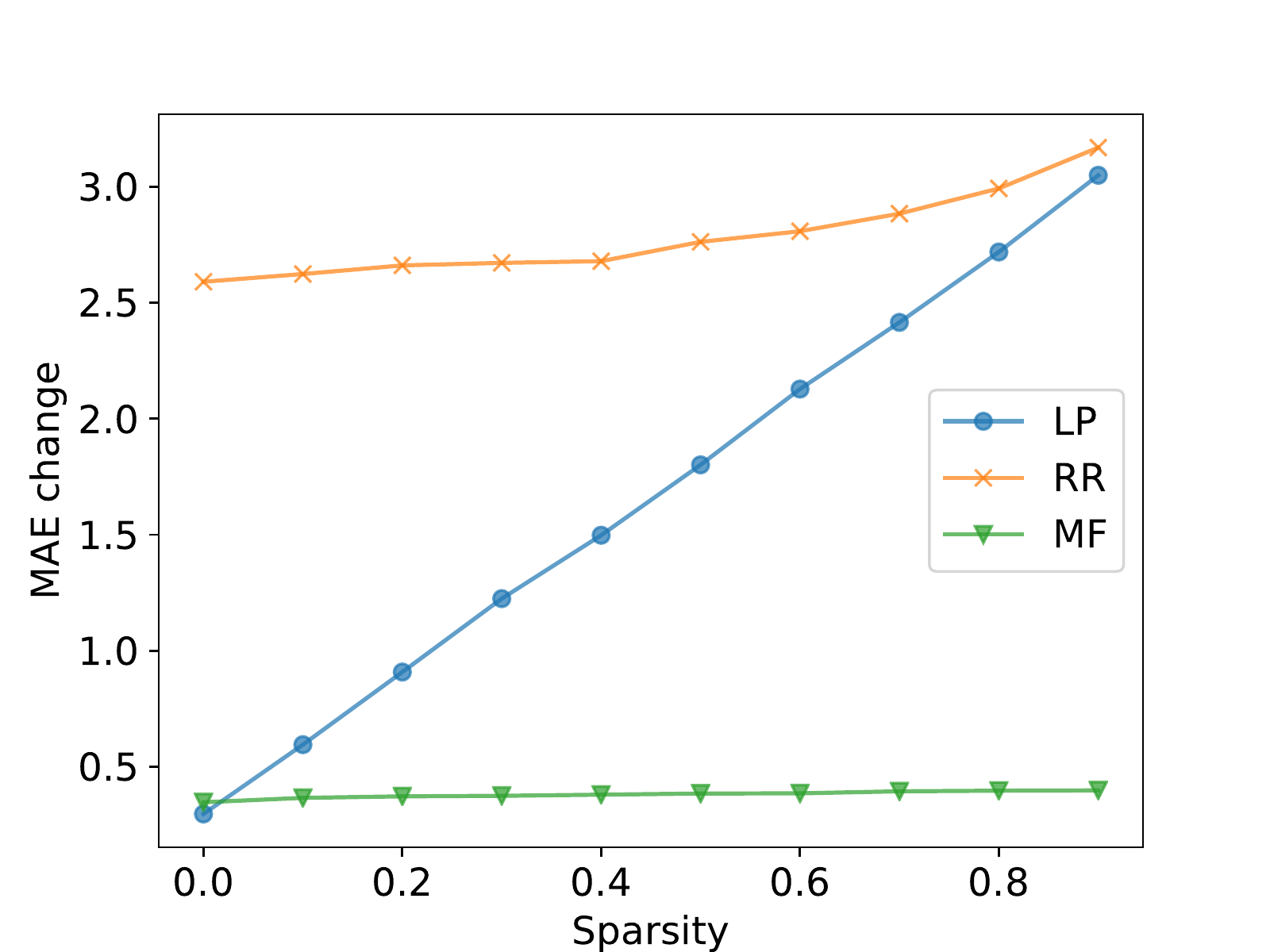}
	\\
	{\scriptsize(a) $\epsilon=0.1$}
	&
	{\scriptsize(b) $\epsilon=1.0$}
    \nop{
	\\
    \includegraphics[width=0.33\textwidth]{./Figures/m2000n200eps0_10skew5_comp4_varyDensity.pdf} 
	&
	\includegraphics[width=0.33\textwidth]{./Figures/m2000n200eps1_00skew5_comp4_varyDensity.pdf}
	&
	\includegraphics[width=0.33\textwidth]{./Figures/m2000n200eps5_00skew5_comp4_varyDensity.pdf}
	\\
	{\scriptsize(d) $\epsilon=0.1$, truth distribution $\mathcal{N}(4,1)$}
	&
	{\scriptsize(e) $\epsilon=1.0$, truth distribution $\mathcal{N}(4,1)$}
	&
	{\scriptsize(f) $\epsilon=5.0$, truth distribution $\mathcal{N}(4,1)$}
}    
\end{tabular}
\vspace{-0.1in}
\caption{\small \label{fig:ae_wrt_density_2000} Accuracy of truth inference w.r.t. different data sparsity (2000 workers and 200 tasks)}
\end{center}
\end{figure}

\noindent{\bf Impact of data sparsity.}
We measure the MAE change of the three mechanisms under various answer sparsity and show the results in Figure \ref{fig:ae_wrt_density_2000}. 
We observe that the MAE change of MF is very small in all settings. More importantly, it is insensitive to answer sparsity. In all cases, the MAE change never exceeds 0.5, even when the sparsity is as high as 0.9. 
This is because MF learns a latent joint distribution between the workers and tasks, and predicts the missing answers accurately. 
In comparison, the utility of LP and RR is sensitive to answer sparsity. 
Overall, the MAE change of these two approaches increases when the sparsity grows. 

\nop{
We also measure the MAE change on the answer distribution where the answers centered at value 4, and show the results in Figure \ref{fig:ae_wrt_density_2000} (d) - (f). 
In Figure \ref{fig:ae_wrt_density_2000} (c) and (f), we observe that MF can only lose to the straw-man approaches when the answers are centered at 4 and the sparsity is small.
When the answers are centered at 0, the original answers distribution are totally different from the uniform distribution. As a consequence, the straw-man approaches result in large errors when replacing the NULL values.
}

\nop{
\Wendy{2. Compare Figure 4 (b) and (e), Why MF loses RR+LP and LP when sparsity is larger than 0.7? Why LP wins RR+LP in Figure 4 (e)? It lost to RR+LP in figure 4 (b).}
By comparing Figure \ref{fig:ae_wrt_density_2000} (b) and (e), the MAE change of LP and RR+LP reduces dramatically when the answers are centered at 4. Its MAE can even be smaller than that of MF when the data is dense. 
Recall that these two approaches firstly replace NULL values by using a uniform distribution, and then adds Laplace noise. Even though the perturbed answers are centered at 4, \Boxiang{This is tricky. I cannot explain it well. I guess the reason is that these approaches fill in the missing cells with $\mathcal{U}(0,9)+Lap(\frac{10}{1})$. And it is different from the way we generate answers $\mathcal{N}(4,1)+\mathcal{N}(0, \sigma_i^2)$, even though both of them are centered at 4. So when the data is sparse, the MAE change can be relatively large.}
}


\nop{
First, we generate synthetic datasets with different density, and measure the truth inference error after applying various perturbation mechanisms. The results are shown in Figure \ref{fig:ae_wrt_density_2000}.
We observe that the error of MF is very small and stable. In all cases, it is no larger than 1.0. This shows the advantage of MF, since it imputes the NULL values according to a joint latent factor space between the workers and minimizes the regularized distortion in Equation (\ref{eq:dp_loss}). 
}

\vspace{-0.2in}
\begin{figure}[!htbp]
\begin{center}
\begin{tabular}{@{}c@{}c@{}}
	\includegraphics[width=0.25\textwidth]{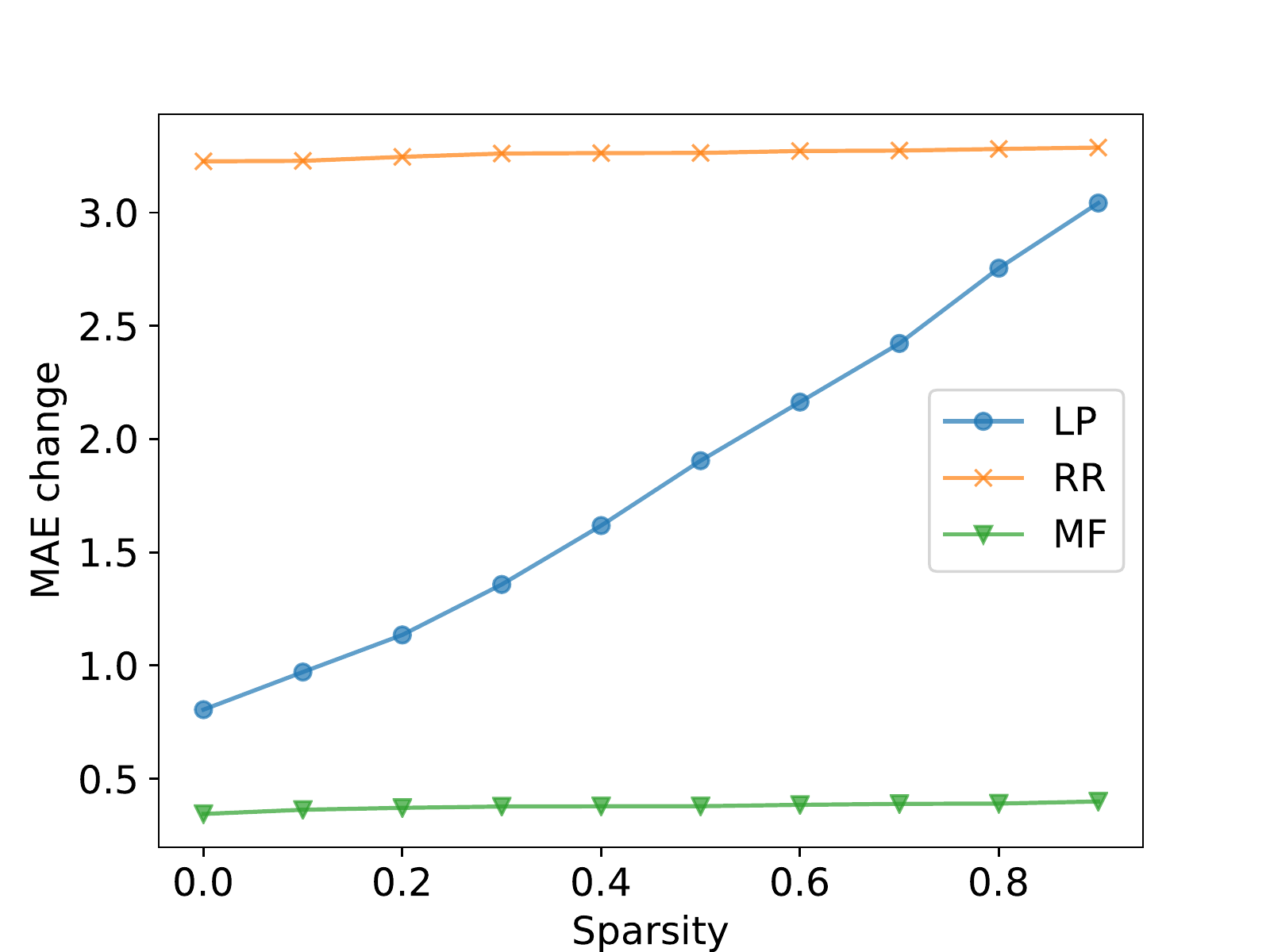} 
	&
	\includegraphics[width=0.25\textwidth]{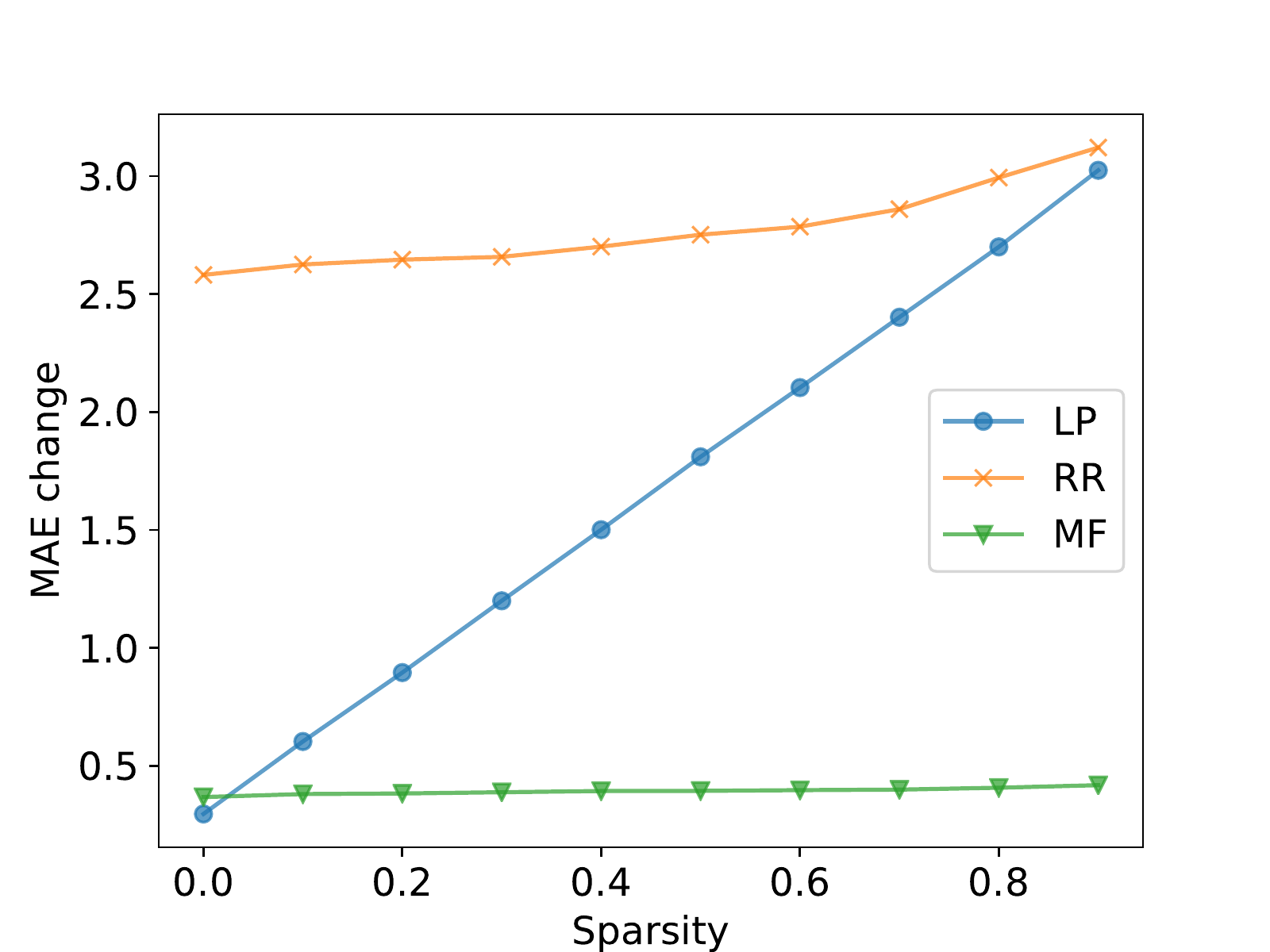}
	\\
	{\scriptsize(a) $\epsilon=0.1$}
	&
	{\scriptsize(b) $\epsilon=1.0$}
    \nop{
	\\
    \includegraphics[width=0.33\textwidth]{./Figures/m10000n1000eps0_10skew5_comp4_varyDensity.pdf} 
	&
	\includegraphics[width=0.33\textwidth]{./Figures/m10000n1000eps1_00skew5_comp4_varyDensity.pdf}
	&
	\includegraphics[width=0.33\textwidth]{./Figures/m10000n1000eps5_00skew5_comp4_varyDensity.pdf}
	\\
	{\scriptsize(d) $\epsilon=0.1$, truth distribution $\mathcal{N}(4,1)$}
	&
	{\scriptsize(e) $\epsilon=1.0$, truth distribution $\mathcal{N}(4,1)$}
	&
	{\scriptsize(f) $\epsilon=5.0$, truth distribution $\mathcal{N}(4,1)$}
    }
\end{tabular}
\vspace{-0.1in}
\caption{\small \label{fig:ae_wrt_density_10000} Accuracy of truth inference w.r.t. different data sparsity (10000 workers and 1000 tasks)}
\end{center}
\vspace*{-0.05in}
\end{figure}

\noindent{\bf Impact of data size.}
We generate the datasets that consist of 10,000 workers and 1,000 tasks. The results for MAE change on the datasets with different sparsity are displayed in 
Figure \ref{fig:ae_wrt_density_10000}. The observations are very similar to those on the small datasets (Figure \ref{fig:ae_wrt_density_2000}). 
We highlight the most interesting observations.
First, when $\epsilon=0.1$, the MAE change of LP  is smaller when the data size is larger. LP requires to add the Laplace noise $Lap(\frac{|\Gamma|}{\epsilon})$ to each answer. The perturbation for each answer is large when $\epsilon=0.1$. However, for the same task, the average of the noise is inversely proportional to the number of answers. Therefore, when there are more answers, the aggregate noise from the workers reduces. We also notice that the MAE change resides in the same range for both small and large datasets. The reason is that the noise added to each answer is independent from the data size. 
We also vary the privacy budget $\epsilon$ on the large datasets and measure the MAE change. The results are similar to the small datasets (Figure \ref{fig:ae_wrt_epsilon_200}). We include the results in the full paper \cite{sun2018truth}.

\nop{
firstly try different size of synthetic datasets to evaluate the Scalability of all three approaches. Here, we choose two datasets: (1) 2000 workers, 200 tasks; (2) 10000 workers, 500 tasks. The two synthetic datasets have 50\% of high-quality workers and 50\% of low-quality workers. Figure \ref{fig:ae_wrt_density_2000} shows the accuracy of smaller dataset and figure \ref{fig:ae_wrt_density_10000} the accuracy of the larger one. We can see that all three approaches are not significantly affected by the size of input. However, comparing sub-figure (a), (b) and (c) on both datasets, we can see that Laplace perturbation (LP) is very sensitive to privacy budget $\epsilon$. When $\epsilon=0.01$, the MAE of LP on both dataset is much larger than the other two approaches, while random response (RR) and matrix factorization (MF) have quite stable MAE because the privacy budge is beyond decreasing segment of them. When $\epsilon$ gets larger, the MAE of LP is gets closer to MF as they share the same very basic ideas: generate noise following Laplace distribution.

Another observation is that MF and LP is more sensitive to the answer density than RR. Because MF and LP are both using Laplace noise and RR uses completely different mechanism. However, when density is larger than 0.5, MF and RR (when $\epsilon\ge 0.1$) can have better accuracy than RR.

Then for the three real-world datasets with different size, since the density of the real-world datasets are fixed, evaluated the mean absolute error w.r.t. different privacy budget $\epsilon$. The results are shown in figure \ref{fig:ae_wrt_epsilon_realworld}. The overall observation remains the same as synthetic datasets: the Laplace perturbation approach is very sensitive to privacy budget $\epsilon$ when we vary it in interval $[0.1,1.0]$, and random response and our matrix factorization approach have quite stable and high accuracy because the privacy budget is beyond the decreasing segment of the two approaches. When we increase privacy budget $\epsilon$, the performance of Laplace perturbation is getting closer to matrix factorization because they share the same very basic perturbation idea. However, in figure \ref{fig:ae_wrt_epsilon_realworld} (c), in the relevance, RR has a better accuracy than MF and LP. \Haipei{???}
}

\nop{
\subsection{Impact of Different Perturbation Approaches on Worker Distribution}
\label{section:perturbation_impact_distribution}

In this section, explore how three approaches affect the original worker quality distribution that inferred from algorithm \ref{alg:estimate}. As we defined in section \ref{section:truth_inference}, we assign equal quality $\frac{1}{n}$ to all workers and keep the sum of qualities to be 1 during the iteration. This will make the quality difficult to show because they are very small when $n$ is large. Here, we let all the worker qualities multiplied by $n$. Therefore, that the workers will have a quality of 1 at the beginning. When algorithm ends, we define the workers with quality larger or equal to 1 as good quality workers and the workers with quality smaller than 1 as bad quality workers.

Before we start, we need to know if the worker quality that inferred by algorithm \ref{alg:estimate} is accurate. Here, do experiments on synthetic dataset with 2000 workers and 200 tasks, which contains the ground truth of worker quality. We use the dataset with a density of 0.8, half high-quality workers and half low-quality workers. The result is shown in figure \ref{fig:symthetic_quality} (a), in which the x-axis is the relative quality values and y-axis is the percentage in each slice. We can see that algorithm \ref{alg:estimate} has successfully and accurately inferred the worker quality distribution, i.e., 50\% of the workers has quality less than 1 and another 50\% is larger than 1.
We can also estimate how the good workers dominate the final result by using following 

After that, we apply three approaches on the dataset and evaluated how they impact the worker quality distribution. The result are shown in \ref{fig:synthetic_quality} (b)-(d)

We can see that LP has destroyed the original worker distribution, and on the contrary, RR and MP somehow preserved the original distribution to some extend. This result explain the order of accuracy in figure \ref{fig:ae_wrt_density_2000} (b).
}

\nop{
\begin{figure*}[!htbp]
\begin{center}
\begin{tabular}{@{}c@{}c@{}c@{}c@{}}
    \includegraphics[width=0.25\textwidth]{./Figures/quality_websearch.pdf} 
	&
	\includegraphics[width=0.25\textwidth]{./Figures/quality_websearch_distribution_eps0_10_.pdf}
	&
	\includegraphics[width=0.25\textwidth]{./Figures/quality_websearch_distribution_eps1_00_.pdf}
    &
	\includegraphics[width=0.25\textwidth]{./Figures/quality_websearch_distribution_eps5_00_.pdf}
    \\
    {\scriptsize(a) Original quality distribution}
	&
	{\scriptsize(b) Perturbed quality distribution}
	&
	{\scriptsize(c) Perturbed quality distribution}
    &
    {\scriptsize(d) Perturbed quality distribution}
    \\
	&
	{\scriptsize privacy budget $\epsilon=0.1$}
	&
	{\scriptsize privacy budget $\epsilon=1.0$}
    &
    {\scriptsize privacy budget $\epsilon=5.0$}
\end{tabular}
\vspace{-0.1in}
  \caption{\small \label{fig:websearch_quality} Worker quality distributions of web search dataset, 34 workers, 177 tasks, density=0.294118, privacy budget $\epsilon=0.1$}
\end{center}
\end{figure*}

\begin{figure}[!htbp]
\begin{center}
\begin{tabular}{@{}c@{}c@{}c@{}}
    \includegraphics[width=0.25\textwidth]{./Figures/quality_websearch.pdf} 
	&
	\includegraphics[width=0.25\textwidth]{./Figures/quality_adult_distribution_eps0_10_.pdf}
	&
	\includegraphics[width=0.25\textwidth]{./Figures/quality_adult_distribution_eps1_00_.pdf}
    &
	\includegraphics[width=0.25\textwidth]{./Figures/quality_adult_distribution_eps5_00_.pdf}
    \\
    {\scriptsize(a) Original quality distribution}
	&
	{\scriptsize(b) Perturbed quality distribution}
	&
	{\scriptsize(c) Perturbed quality distribution}
    &
    {\scriptsize(d) Perturbed quality distribution}
    \\
	&
	{\scriptsize privacy budget $\epsilon=0.1$}
	&
	{\scriptsize privacy budget $\epsilon=1.0$}
    &
    {\scriptsize privacy budget $\epsilon=5.0$}
\end{tabular}
\vspace{-0.1in}
  \caption{\small \label{fig:adultcontent_quality} Worker quality distributions of adult content dataset, 825 workers, 11,040 tasks, density=0.006334 \Boxiang{I think we can remove Figure 8, 9 and 10. The only explanation we can give is that MF is accurate in imputation and perturbation, so that it preserves the worker quality distribution somehow. While the other approaches simply destroy it.}}
\end{center}
\end{figure}

\begin{figure*}[!htbp]
\begin{center}
\begin{tabular}{@{}c@{}c@{}c@{}c@{}}
    \includegraphics[width=0.25\textwidth]{./Figures/quality_websearch.pdf} 
	&
	\includegraphics[width=0.25\textwidth]{./Figures/quality_relevance_distribution_eps0_10_.pdf}
	&
	\includegraphics[width=0.25\textwidth]{./Figures/quality_relevance_distribution_eps1_00_.pdf}
    &
	\includegraphics[width=0.25\textwidth]{./Figures/quality_relevance_distribution_eps5_00_.pdf}
    \\
    {\scriptsize(a) Original quality distribution}
	&
	{\scriptsize(b) Perturbed quality distribution}
	&
	{\scriptsize(c) Perturbed quality distribution}
    &
    {\scriptsize(d) Perturbed quality distribution}
    \\
	&
	{\scriptsize privacy budget $\epsilon=0.1$}
	&
	{\scriptsize privacy budget $\epsilon=1.0$}
    &
    {\scriptsize privacy budget $\epsilon=5.0$}
\end{tabular}
\vspace{-0.1in}
  \caption{\small \label{fig:relevance_quality} Worker quality distributions of relevance finding dataset, 766 workers, 19,902 tasks, density=0.009859}
\end{center}
\end{figure*}
}



\begin{figure}[!htpb]
\begin{center}
\begin{tabular}{@{}c@{}c@{}}
	\includegraphics[width=0.25\textwidth]{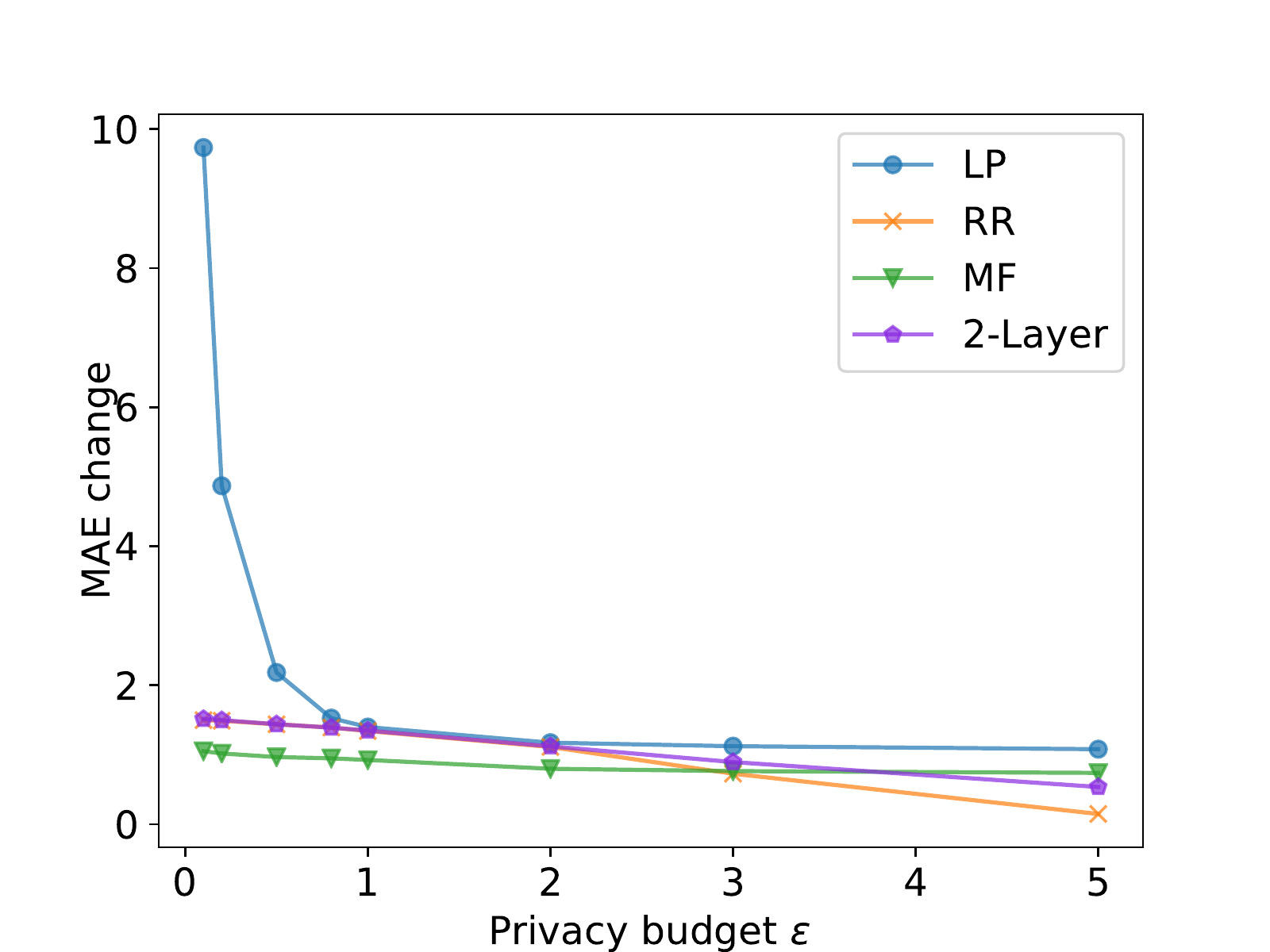} 
	&
	\includegraphics[width=0.25\textwidth]{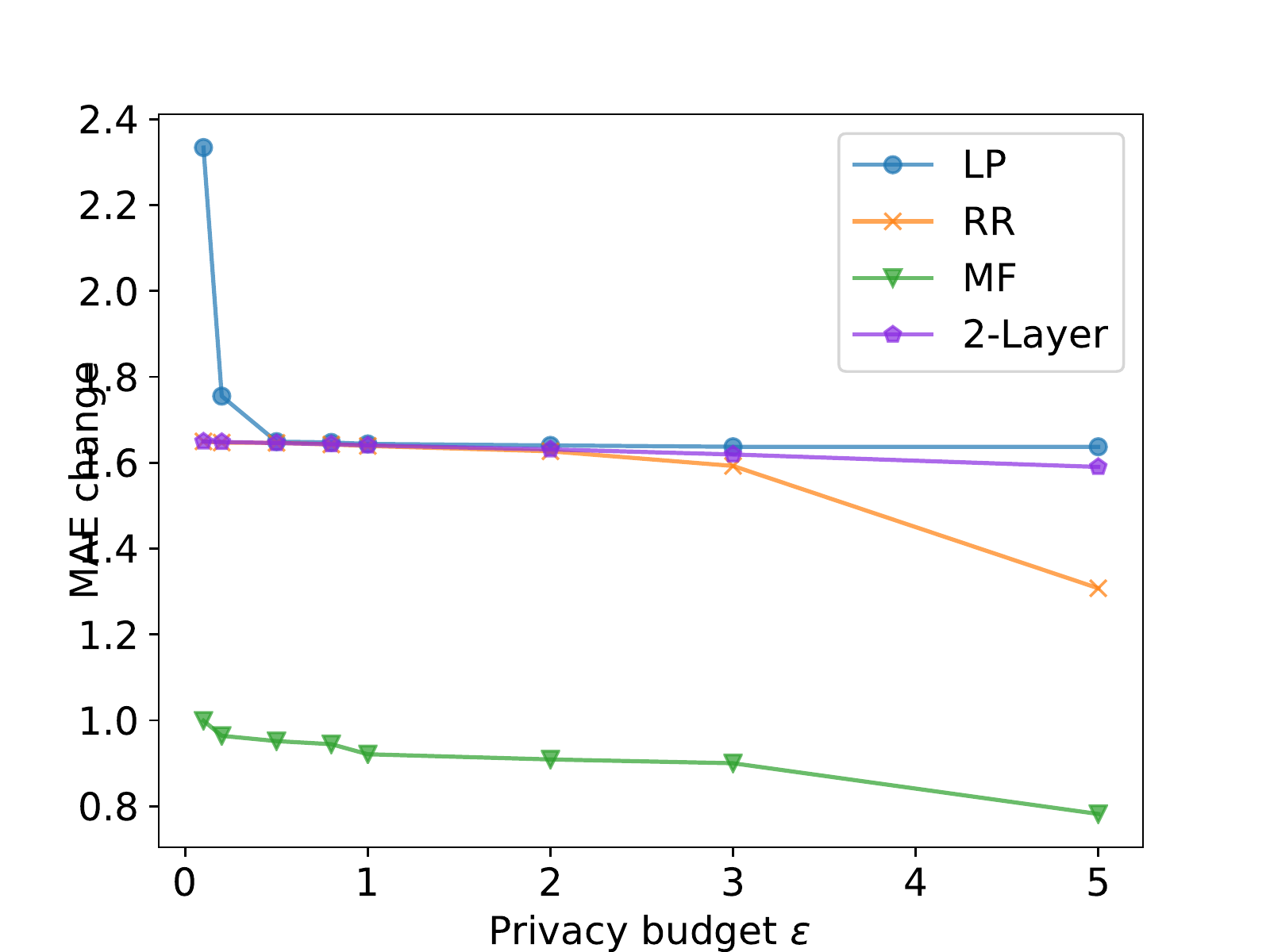}
	\\
	{\scriptsize(a) Web dataset}
	&
	{\scriptsize(b) AdultContent dataset}
\end{tabular}
\vspace{-0.1in}
  \caption{\small \label{fig:ae_wrt_epsilon_realworld} Accuracy of truth inference w.r.t. privacy budget $\epsilon$ on two real-world datasets.}
  \vspace*{-0.1in}
\end{center}
\end{figure}

\vspace{-0.1in}
\subsection{Comparison with Existing Method}
\vspace{-0.1in}
We compare the MAE change of RR, LP, and MF with the 2-Layer approach \cite{li2018an} on the two real-world datasets. We vary the privacy budget $\epsilon$ and evaluate the MAE change of the four approaches. The results are shown in Figure \ref{fig:ae_wrt_epsilon_realworld}. We have the following observations. 
First, similar to the results on the synthetic datasets, our MF approach has the lowest MAE change for most of the cases. 
By matrix factorization, the MF approach fills NULL values with non-NULL values that follow the current  answer distribution. Therefore, the inferred truths from the perturbed answers are close to the inferred truths of the original data, which leads to a small MAE change. In contrast, by LP, RR, and 2-layer approaches, the NULL values are replaced with values in the range [0,4] with mean 2. Since the dataset is very sparse, such NULL value replacement shifts the inferred truth closer to 2 (compared with the inferred truth 0 from the original dataset), and thus further away from the ground truth. Therefore, the MAE change of these three approaches is larger than that of MF. 
Second, the 2-Layer approach shows comparable MAE change as RR. This is because it decides the privacy budget $\epsilon_i$ of each worker $W_i$ by following a uniform distribution with mean of $\epsilon$ and then perturbs the answers by following randomized response. Such data perturbation scheme introduces similar noise as RR, especially when the number of workers is large.
Third, in the {\em AdultContent} dataset, the MAE change of the RR approach has a sharp drop when $\epsilon$ gets larger than 3. The reason is that, 
when $\epsilon$ gets larger than 3, RR keeps the NULL values unchanged with high probability (approximately $80\%$). The inferred truth on such perturbed dataset is very close to that from the original dataset, which leads to smaller MAE change. 



\nop{
\subsection{Summary of Experiments}
To summarize, our experiments show that, in general, the MF approach delivers better utility than the other approaches, especially when: (1) the datasets are very sparse; and (2) the privacy budget is small. 
}



\vspace{-0.1in}
\section{Related Work}
\label{sc:related}
\vspace{-0.1in}
\nop{
\noindent{\bf Data Collection with Privacy Protection.}
To et al. \cite{to2015privgeocrowd} propose a approach called PrivGeoCrowd to protect the workers' privacy in spatial crowdsourcing that can be inferred by knowing the exact location. Unlike our work, the workers in their setting reports their privacy data to a Cell Service Provider without perturbation, i.e., they are using centralized differential privacy.
Erlingsson et al. \cite{erlingsson2014rappor} design a random response-based approach named RAPPOR. In their setting the clients are assumed to answer each question, i.e., they do not consider sparsity.
To preserve the workers' privacy, \cite{wang2016incentive} use cryptographic approaches, rather than differential privacy, to preserve the workers' privacy in mobile crowdsourcing systems; \cite{liu2012cloud,keren2014monitoring,hsu2012distributed,chan2012differentially} restricts the amount of communications between workers and the untrusted data collector;
\cite{hsu2012distributed} assumes there is a trustable third party (TTP) and \cite{de2011short} uses anonymous data collection to protect the workers' privacy.
}

Differential privacy (DP) \cite{Dwork06differentialprivacy} is a privacy definition that requires the output of a computation should not allow inference about any record’s presence in or absence from the computation's input. 
However, the classic DP model assumes the data curator is trusted. Therefore, it cannot be adapted to the crowdsourcing model in which the data curator is untrusted. 
Local differential privacy (LDP) has recently surfaced as a strong measure of privacy in contexts where personal information remains private even from data analysts.
LDP has been widely used in practice.
\cite{erlingsson2014rappor} propose RAPPOR to collect and analyze users' data without privacy violation. RAPPOR is built on the concept of randomized response, which was proposed \cite{warner1965randomized} initially for collection of survey answers. 
RAPPOR has been 
deployed in Google's Chrome Web browser for reporting the statistics of users' browsing history \cite{rapporgoogle}. 
However, RAPPOR is limited to simple data analysis such
as counting. It cannot deal with aggregates and sophisticated analysis such as classification and clustering \cite{shin2018privacy}. 
Apple applies LDP to learn new words typed by users and identify frequently used Emojis while retaining user privacy \cite{thakurta2017learning}. 
Microsoft adapts LDP to its collection process of
a variety of telemetry data \cite{ding2017collecting}.
A series of work \cite{bassily2017practical,bassily2015local,qin2016heavy,wang2017locally} aim to make LDP practical by improving the error of frequency estimation and heavy hitters on the data perturbed by LDP.
Chen et al. \cite{chen2016private} design a framework for aggregation of private location data with LDP. Since the same location is not equally sensitive to multiple users, the authors develop a personalized LDP protocol based on randomized response.
Qin et al. \cite{qin2017generating} designed a protocol to reconstructed decentralized social networks by collecting neighbor lists from users while preserving LDP. 
\nop{
considers social graphs, and design the algorithms to protect the private neighbor lists with LDP.}
 Nyugen et al. \cite{nguyen2016collecting} design an LDP mechanism named Harmony to collect and analyze data from users of smart devices. Harmony supports both basic statistics (e.g., mean and frequency
estimates), and complex machine learning tasks (e.g., linear
regression, logistic regression and SVM classification).
Other works on LDP \cite{erlingsson2014rappor,fanti2016building,qin2016heavy,wang2017locally} only focus on simple statistical aggregation functions, such as identifying heavy hitters and count estimation. None of these works consider truth inference as the utility function. 
\nop{
\noindent{\bf Sparse Data Handling.}
So far there is not much work that considers how to handle sparsity in differential privacy mechanisms, especially the sparsity is presented as NULL value. 
Cormode et al. \cite{cormode2011differentially} add differentially private noise on the sparse dataset and filter negative entries. Then they generate a compact summary of the noisy data by sampling. Unlike our work, they consider 0 values in the original dataset as the missing data rather than NULL.
Nikolov et al. \cite{nikolov2013geometry} present a geometric approach to achieve differential privacy on sparse database. They firstly achieve give $(\epsilon,\delta)$-differential privacy by using orthogonal decomposition and Gaussian noise, then the give $\epsilon$-differential privacy approach by using approximation of K-norm distribution. However, different from our work, their approaches are designed for aggregate linear queries over histograms such as contingency tables and range queries, not specific cells in the database.
So as in \cite{venanzi2014community,misra2014crowdsourcing}.
}

The problem of protecting worker privacy in crowdsourcing systems has received much attention recently. 
Kairouz et al. \cite{kairouz2014extremal} investigate the trade-off between worker privacy and data utility as a constraint optimization problem, where the utility is measured as information theoretic statistics such as mutual information and divergence. They prove that solving the privacy-utility problem is equivalent to solving a finite dimensional linear program. Their utility measures are different from ours (truth inference). 
Ren et al. \cite{ren2018textsf} consider the problem of publishing high-dimensional crowdsourced data with LDP. The server aggregates the data (in the format of randomized bit strings) from individual users and estimates their join distribution. Thus their utility function (i.e., estimation of joint distribution) is fundamentally different from ours. 
Among all the previous works, probably \cite{li2018an} is the most relevant to ours. The authors propose a two-layer perturbation mechanism based on randomized response to protect worker privacy, while providing high accuracy for truth discovery. However, \cite{li2018an} assumes that the data is complete. 
B{\'e}ziaud et al. \cite{beziaud2017lightweight} protects worker profiles by allowing the workers to perturb their profiles locally by using randomized response. They do not consider truth inference as the utility function. 

One research direction that is relevant to our work is to apply DP/LDP on recommender systems. 
The data-obfuscation
solutions (e.g. \cite{canny2002collaborative,shen2014privacy,shen2016epicrec}) adapts DP to the recommender systems and rely on adding noise to the original
data or computation results to restrict the information leakage
from recommender outputs. Shin et al. \cite{shin2018privacy} focus on the LDP setting, aiming to protect both rating and item privacy for the users in recommendation systems. They develop a new matrix factorization algorithm under LDP. In specific, gradient perturbation is applied in the iterative factorization process. To reduce the perturbation error, the authors adopt random projection for matrix dimension reduction.

\vspace{-0.1in}
\section{Conclusion}
\label{sc:conclu}
\vspace{-0.1in}
In this paper, we consider the problem of privacy preserving Big data collection under the crowdsourcing setting, aiming to 
protecting worker privacy with LDP guarantee while providing highly accurate truth inference results. 
We overcome the weakness of two existing LDP approaches (namely  randomized response and noise perturbation) due to worker answer sparsity, and design a new LDP matrix factorization method that adds perturbation on objective functions. 

In the future, we plan to continue the study of privacy protection for crowdsourcing systems. 
We plan to design the privacy preservation method for other truth inference algorithms, for example, majority voting based methods \cite{tian2015max}.
We also plan to investigate how to protect task privacy, i.e., the sensitive information in the tasks,  while allow the workers to provide high-quality answers.

\bibliographystyle{abbrv}
{
\bibliography{bib}  
}
\newpage

\end{document}